%% file: main.tex
\documentclass[runningheads]{llncs}  %

\input{preamble.tex}

\title{Optimal Runtime Verification of Finite State Properties over Lossy Event Streams}
\titlerunning{Optimal RV over Lossy Event Streams}

\author{Peeyush Kushwaha\inst{1} \and Rahul Purandare\inst{1}\and Matthew B. Dwyer\inst{2}}
\authorrunning{P. Kushwaha et al.}
\institute{IIIT-Delhi, New Delhi, India 
\email{\{peeyush16254,purandare\}@iiitd.ac.in}\\
\and University of Virginia, Charlottesville, USA
\email{matthewbdwyer@virginia.edu}}

\begin{document}

\maketitle

\input{abstract.tex}

\keywords{Runtime Verification \and Finite State Properties \and Optimization}

\input{intro.tex}

\input{overview.tex}

\input{basic-definitions.tex}

\input{incremental-filters.tex}

\input{more-examples.tex}

\input{emperical-study.tex}

\input{approximate-alternate-monitors.tex}

\input{related-work.tex}

\input{conclusion.tex}

\bibliography{references1,references2}
\bibliographystyle{ieeetr}
    \newpage
    \input{appendix.tex}
\end{document}

%% file: preamble.tex
\usepackage[english]{babel}
\usepackage[utf8]{inputenc}

\usepackage{amsmath}
\usepackage{array}
\usepackage{mathtools} %
\usepackage{graphicx}
\usepackage{todonotes} 
\presetkeys{todonotes}{inline,backgroundcolor=white,bordercolor=orange}{}
\usepackage{braket} 
\usepackage{hyperref}
\usepackage{calc}
\usepackage{parskip} %
\usepackage[shortlabels]{enumitem} %
\usepackage{soul} %
\usepackage{amssymb}
\usepackage{algorithm}
\usepackage[noend]{algpseudocode} %
\usepackage{caption}
\usepackage{subcaption}
\usepackage{multirow}
\usepackage{xpatch} %
\usepackage[capitalize,nameinlink]{cleveref} %
\usepackage{varwidth} %
\crefformat{section}{#2\S#1#3}
\usepackage{braket}
\usepackage{commath}
\usepackage{tabularx}

\usepackage{wrapfig}

\newcommand{\tset}[1]{\{#1\}}

\newcommand{\sst}{\text{ such that }}

\newcommand{\FR}[1]{\mathcal{F}_{#1}}
\newcommand{\FRinv}[1]{\mathfrak{C}}
\newcommand{\mon}[1]{M_{#1}}
\newcommand{\defeq}{=_{def}}

\newcommand{\proofinappendix}{\hfill$\blacksquare$}
\renewcommand{\abs}[1]{{\left|#1\right|}}

\DeclarePairedDelimiter{\floor}{\lfloor}{\rfloor}
\DeclarePairedDelimiter{\ceil}{\lceil}{\rceil}

\usetikzlibrary{positioning,automata,calc,shapes.geometric,arrows,fit} 

\tikzset{%
 initial text={},
 double distance=1.5pt
}
\makeatletter
\def\old@comma{\ }
\catcode`\~=13
\def~{%
  \ifmmode%
    \old@comma\allowbreak%
  \else%
    \old@comma%
  \fi%
}
\makeatother

%% file: abstract.tex
\begin{abstract}

Monitoring programs for finite state properties is challenging due to high memory and execution time overheads it incurs. Some events if skipped or lost naturally can reduce both overheads, but lead to uncertainty about the current monitor state. In this work, we present a theoretical framework to model these lossy event streams and provide a construction for a monitor which observes them without producing \textit{false positives}. The constructed monitor is optimally sound among all complete monitors. We model several loss types of practical relevance using our framework and provide construction of smaller approximate monitors for properties with a large number of states.
\end{abstract}

%% file: intro.tex
\section{Introduction}

\label{sec:introduction}
Monitoring the execution behavior of software goes back to the dawn of
programming and is a standard practice, e.g., through logging, programmer inserted print statements, and assertions. 
In the late 1990s, researchers began to explore the use of formal 
specifications to define run-time monitors~\cite{kim1999ecrts} which brought the expressive power of formal methods to monitoring. 
Such \textit{run-time verification} techniques rely on a set of defined
\textit{events} which denote the occurrence of program behavior relevant to a property specification, e.g., the invocation of a particular method, along with  associated data, e.g., method parameters.
A run-time monitor observes an \textit{event stream} generated by a program execution, 
incrementally updates the \textit{state} of the specified property, and reports a property violation when a violating state is reached.

Run-time verification is attractive because it 
complements sound static verification approaches that
cannot scale to modern software systems.   
The past two decades have witnessed work on scaling run-time 
verification in three dimensions.
The family of specification languages for monitors has been
steadily increasing in number and expressive power, e.g.,~\cite{havelund2002tacas}.
The treatment of high-level abstractions that are present in
modern languages, such as, object identifiers for specifying properties
of class instances, has been addressed, e.g.,~\cite{chen2007oopsla}.
Finally, a range of techniques for reducing the run-time overhead, 
while preserving violation detection, have been
developed, e.g.,~\cite{Bodden10a,Bodden10b,Dwyer07,Purandare10}. 

In this paper, we consider the additional challenge of \textit{lossy event streams} that arises in
the deployment of run-time verification in realistic system contexts such as:
networked and distributed systems where message loss or reordering may 
be inherent, real-time systems which may shed monitoring workloads to 
meet scheduling constraints, or web-based systems with quality-of-service 
guarantees may lead to suppressed monitoring.
In such systems, the original event stream may be perturbed by
dropping events, reordering events, or dropping or corrupting data 
correlated with events.

 Lossy event streams are problematic for existing run-time verification 
approaches since treating a lossy stream the same as the original stream 
may lead to missing a property violation or falsely declaring a violating execution. %
Lossy streams do not, in general, permit the same degree of precision as 
the original stream. However, as we demonstrate, run-time 
verification frameworks can be adapted
to effectively bound the impact of the loss on verification results.

The paper makes foundational contributions to runtime verification by (a) \textit{defining an expressive framework for modeling lossy event streams}, (b) \textit{developing techniques for synthesizing provably complete and optimal verification monitors under those models}.  Importantly, these results preclude the need for additional theory development for individual loss types and set the stage for more applied work and tool development.
The applicability and utility of these methods are demonstrated by (c) \textit{formulating a collection of diverse loss models} and, (d) \textit{evaluating the ability of the methods to detect property violations in the the presence of losses in practice}. 

Moreover, we discuss how the event losses in current literature are specific instances of our generalized framework in \cref{sec:more-examples} and \cref{sec:related}.

In the next section, we  provide an overview of our solution. After introducing notation and basic definitions, in \cref{sec:filters} we formalize the loss model and show how verification monitors can be constructed that are complete and optimal for that model. \cref{sec:more-examples} presents example instantiations of the framework that highlight its range.
We describe related work in \cref{sec:related}.

%% file: overview.tex
\section{Overview}
\label{sec:motivation}

\begin{figure}
\centering
\begin{subfigure}{0.3\textwidth}
\centering
\input{tikzpictures/safeiter-property.tex}
\caption{SafeIter Property\label{fig:safeiter-property}}
\end{subfigure}%
\hfill
\begin{subfigure}{0.3\textwidth}
\centering
\input{tikzpictures/safeiter-filter.tex}
\vspace*{-3mm}
\caption{Loss Filter\label{fig:safeiter-filter}
}
\end{subfigure}%
\hfill
\begin{subfigure}{0.3\textwidth}
\centering
\input{tikzpictures/safeiter-transformed.tex}
\vspace*{-5mm}
\caption{Alternate Property\label{fig:safeiter-transformed}
}
\end{subfigure}
\label{fig:leading-safeiter-example}
\caption{Safe Iterator -- after creating (c) an iterator no updates (u)
are permitted so long as next (n) elements remain to be accessed.}
\end{figure}
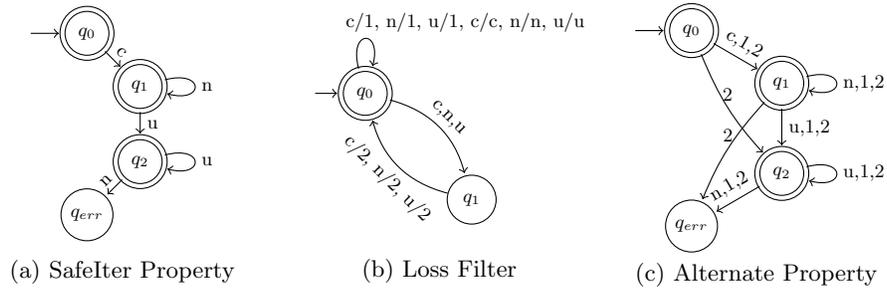

We illustrate the problem of loss in monitoring stateful properties
by way of example and introduce the key insights behind the
techniques we develop to address the problem.

Safety properties for run-time monitoring can be modelled using deterministic finite-state automata (DFA). An event is represented by a \emph{symbol} and an event stream by a \emph{string} of symbols. 
Fig.\ref{fig:safeiter-property} shows the DFA for the 
\textit{safe iterator} property which states that modification of 
a collection during iteration is not permitted.  
The DFA is expressed over the alphabet $\{c, n, u\}$ denoting
creation of an iterator, accessing the next element in the iterator,
and updating the collection being iterated.

The state $q_{err}$ is a sink state (self transitions ommitted for brevity)
denoting the violation of the property.  All violating strings include
a subsequence $\ldots u n \ldots$ indicating that an update was performed
prior to accessing the next element.
Statements that are free of such a subsequence end at one of the three
accept states and are non-violating.

Loss may come in different forms.
For example, symbols in a string may be erased (e.g., $c \rightarrow \epsilon$), reordered (e.g., $n u \rightarrow u n$), or be modeled with only
partial information (e.g., their count $n n n \rightarrow 3$).

To illustrate, we consider the case where symbols are dropped from
the string, but the number of dropped symbols is recorded.
This type of loss could be introduced intentionally as a means of
mitigating excessive runtime overhead in monitoring, while preserving
fault detection capability.
Consider a string $c n n u n$ where this loss is applied to the 2nd
and 3rd symbols -- we model the resulting string as $c 2 u n$.
This could represent 4 possible strings with a prefix of $c$,
followed by an element of $\{ u, n\}^2$, and a suffix of $u n$.
For longer strings where sequences of length $k$ are lost,
the combinatorics of their possible replacements $\{ u, n \}^k$ 
make it intractable to consider all of the possibilities.
Despite this, the structure of Fig.~\ref{fig:safeiter-property} dictates 
that any string of the form $c k u n$ violates the property, 
thereby illustrating that even with loss it is possible to perform 
accurate monitoring.

We formalize the intuition above in a \textit{loss model} that
maps symbols from the property to an alternative symbol set.
For example, the loss model described above is defined by the mapping
$\{ u, n \}^k \mapsto k$. 
In \S\ref{sec:filters}, we show that all mappings of interest are a 
restricted class of relations on strings called rational relations. 
Rational relations can be represented by non-deterministic 
finite-state transducers (NFTs). 
An NFT maps an alphabet, $\Sigma$, to an alternative alphabet, $\Sigma_a$. 
Fig.~\ref{fig:safeiter-filter} shows the NFT with the mapping 
for alternate symbols $1$ and $2$ that lose the identify of 
symbols in a subsequence but retain the length of the subsequence.
Then $c 2 u n$ represents the loss of identity of the 2nd and 3rd symbols of
any string of length 5.

\begin{wrapfigure}[12]{r}{0.35\textwidth}
\vspace*{-10mm}
\begin{tabular}{ |c|c|c| } 
\hline
   & String & State(s) \\ 
\hline\hline
  $O_1$ & cnnunnun & $q_{err}$\\ 
\hline
 $F_1$ & 2nun2n  & $\set{q_{err}} $\\ 
   & cn1unnun  & $\set{q_{err}} $\\ 
   & cn2nn2  & $\set{q_2, q_{err}}$ \\ 
 \hline\hline
  $O_2$ & cnnuu & $q_2$\\ 
\hline
 $F_2$ & c2uu & $\set{q_2, q_{err}} $\\ 
   & cnnu1  & $\set{q_2, q_{err}} $\\ 
   & 2n2  & $\set{q_1, q_2, q_{err}}$ \\ 
 \hline\hline
\end{tabular}
\caption{Filtered Strings\label{fig:safeiter-table}}
\end{wrapfigure}
Retaining partial information about an event string 
might be insufficient to conclude that a violation 
occurred (or did not occur). 
We report a violation only when the partial information is 
sufficient to conclude that there must be a violation -- such
monitoring is \textit{complete} since it never reports a false violation.
Consider the original string ($O_1$) in Fig.~\ref{fig:safeiter-table}
and the set of 3 filtered strings ($F_1$) induced by the NFT in
Fig.~\ref{fig:safeiter-filter}.
Tracing through Fig.~\ref{fig:safeiter-property} on the first two
filtered strings by interpreting $1$ and $2$ as any individual or 
pair of symbols, respectively, leads only to the
error state -- since they preserve the fact that an $n$ follows a $u$.
These strings would be reported as violations.
On the other hand, the string $c n n u u$ ($O_2$) is non-violating and of the set of 3 filtered strings ($F_2$) none reach \textit{only} the $q_{err}$ state.  Completeness assures that no filtered string can have $u$ followed by a $n$ and monitor in Fig.~\ref{fig:safeiter-filter} won't report a false violation.

Assuring completeness in violation reporting means, however, that
the reporting of some violations may be missed. 
For example, the third filtered string for $F_1$ suppresses all
$u$ making it impossible to definitively conclude that the observed
string is a violation.
Our goal is to report violations on as many strings of 
alternate symbols as possible while maintaining completeness.
We refer to this as \textit{optimal lossy monitoring} and present its formulation in \cref{sec:filters} and evaluate its tolerance to loss in practice in \cref{sec:evaluation}.

As detailed in \cref{sec:filters}, monitoring under a loss model, expressed as an NFT, is achieved by
transforming the property of interest to an NFA with transitions on
the symbols in the NFT's output alphabet, $\Sigma_a$.
Fig.~\ref{fig:safeiter-transformed} shows the alternate property
for the property in Fig.~\ref{fig:safeiter-property}
transformed by the NFT in Fig.~\ref{fig:safeiter-filter}.
Transitions on the alternative symbols, $\{ 1, 2 \}$, give rise to
non-determinism.
The transition function for a property DFA, 
$\delta : S \times \Sigma \rightarrow S$, or an alternative property NFA,
$\delta_a : S \times \Sigma_a \rightarrow 2^S$, naturally
lift to strings and sets of strings.  
The correctness criterion for alternate properties requires that
$\forall \sigma \in \Sigma^* : \delta(\sigma) \subseteq \delta_a(f(\sigma))$, 
and completeness demands that errors are reported only when
$\delta_a(f(\sigma)) \subseteq q_{err}$.

Our formalization of symbol loss is general.  It can 
be used to addresses the notion of \textit{natural} loss, 
e.g., where environmental factors result in a symbol being dropped from an 
event stream. 
Loss can also be \textit{induced} artificially as a means of suppressing events or data associated with events to reduce
monitoring overhead, which the framework accommodates 
naturally.   We justify the breadth of applicability of our
approach by demonstrating that it accommodates existing loss 
types in the literature \cite{Falzon13,Purandare10,Dwyer08} in \cref{sec:more-examples}.

%% file: tikzpictures/safeiter-property.tex
\usetikzlibrary{positioning,automata,calc,shapes.geometric,arrows,fit} 
\begin{tikzpicture}[shorten >=1pt,node distance=1,on grid]
  \tikzset{every node/.style={scale=0.8}}
  \node[state,accepting,initial]    (q_0)                     {$q_0$}; 
  \node[state,accepting]            (q_1)     [below right=of q_0]  {$q_1$}; 
  \node[state,accepting]            (q_2)     [below=of q_1]  {$q_2$}; 
  \node[state]                      (err)     [below left=of q_2]  {$q_{err}$}; 
  
  \path[->] (q_0)     edge                node [above,sloped] {c}   (q_1)
            (q_1)     edge  [loop right]  node          {n} ()
            (q_1)     edge                node [right] {u} (q_2)
            (q_2)     edge  [loop right]  node [right] {u} ()
            (q_2)     edge                node [above,sloped] {n} (err);
\end{tikzpicture}

%% file: tikzpictures/safeiter-filter.tex
\usetikzlibrary{positioning,automata,calc,shapes.geometric,arrows,fit} 
\begin{tikzpicture}[shorten >=1pt,node distance=1.3,on grid, scale=0.8]
  \tikzset{every node/.style={scale=0.8}}
  \node[state,accepting,initial]    (q_0)                      {$q_0$}; 
  \node[state]                      (q_1)     [below right=2 of q_0] {$q_1$}; 
  
  \path[->] (q_0)     edge [loop above]   node [above right= 0 and -10pt]      {c/1, n/1, u/1, c/c, n/n, u/u} (q_0)
            (q_0)     edge [bend left=30] node [above,sloped]      {c,n,u}                (q_1)
            (q_1)     edge [bend left=30] node [below,sloped]      {c/2, n/2, u/2}           (q_0)
            ;
\end{tikzpicture}

%% file: tikzpictures/safeiter-transformed.tex
\usetikzlibrary{positioning,automata,calc,shapes.geometric,arrows,fit} 
\begin{tikzpicture}[shorten >=1pt,node distance=1,on grid]
  \tikzset{every node/.style={scale=0.8}}
  \node[state,accepting,initial]    (q_0)                     {$q_0$}; 
  \node[state,accepting]            (q_1)     [below right= 0.7 and 1.2 of q_0]  {$q_1$}; 
  \node[state,accepting]            (q_2)     [below= 1.2 of q_1]  {$q_2$}; 
  \node[state]                      (err)     [below left= 0.7 and 1.2 of q_2]  {$q_{err}$}; 
  
  \path[->] (q_0)     edge                node [above,sloped] {c,1,2}   (q_1)
            (q_0)     edge  [bend right=10]   node [above] {2}   (q_2)
            (q_1)     edge  [loop right]  node          {n,1,2} ()
            (q_1)     edge                node [right] {u,1,2} (q_2)
            (q_1)     edge  [bend right=10]  node [above] {2}   (err)
            (q_2)     edge  [loop right]  node [right] {u,1,2} ()
            (q_2)     edge                node [above,sloped] {n,1,2} (err);
\end{tikzpicture}

%% file: basic-definitions.tex
\section{Basic Definitions}
\label{sec:basic-definitions}

Notation: $x \prec y$ means string $x$ is a proper prefix of string $y$. A function $f: X \rightarrow Y$ lifted to sets means that $f(S) = \set{f(x) \mid x \in S} \forall S \subseteq X$. Middle dot ($\cdot$) denotes string concatenation. It may also be lifted to sets of strings.  $\times$ denotes the cartesian product of two sets. For a relation $R \subseteq X \times Y$, $R(x) = \set{y \mid xRy}$ and $R^{-1}(y) = \set{x \mid xRy}$. $\Rightarrow \Leftarrow$ denotes a contradiction. $\#x(y)$ denotes number of characters $x$ in string $y$. If we write an element $x \in X$ where a set $S \subseteq X$ is expected, it denotes the singleton set $\set{x} \subseteq X$. $f_{|X|}$ denotes a restriction of the function $f$ to a subset $X$ of its domain. A partition $\mathcal{P}$ of a set $S$ is a set $\set{P_1, P_2, \ldots}$ such that $P_i$ are pairwise disjoint nonempty sets (called equivalence classes) whose union is $A$. A class representative of $P_i$ is a distinguished element in $P_i$. $[s]_\mathcal{P}$ and $rep^\mathcal{P}(s)$ denote equivalence class and class representative of an element $s \in S$ in $\mathcal{P}$. $\blacksquare$ represents a proof that is %
available in the appendix.

Familiarity with regular languages and their properties is assumed. An observance of a \emph{symbol} is called an \emph{event}. A finite set of symbols is called an \emph{alphabet}.  $REG(\Sigma)$ is the set of all regular languages over an alphabet $\Sigma$. $\epsilon$ is the empty string, and $\Sigma_\epsilon$ is the alphabet $\Sigma \cup \set{\epsilon}$. A \emph{trace} is a (possibly infinite) sequence of events, and an \emph{execution} is a finite prefix of a trace. A trace $x$ is a \emph{continuation} of an execution $x'$ if $x' \prec x$.

\subsection{Finite-state Machines and Transducers}

\begin{definition}[Finite Automata]
A finite automaton is a 5-tuple $(Q, \Sigma, \delta, q_0,$ $F)$ where $Q$ is the finite set of states, $\Sigma$ is the alphabet, $q_0 \in Q$ is a specified initial state and $F \subseteq Q$ is the set of final states. A deterministic finite automaton (DFA) has the transition function $\delta: Q \times \Sigma \rightarrow Q$ and a nondeterministic finite automaton (NFA) has the transition function $\delta: Q \times \Sigma \rightarrow 2^Q$.
The transition function $\delta$ is lifted to strings, sets of strings, and sets of states. We call $L(A) = \set{ x \in \Sigma^* \mid \delta(q_0, x) \in F}$ the \emph{language} of the finite automaton.
\end{definition}

\begin{definition}[Nondeterministic Finite-State Transducers (NFTs)]
Defined as a NFA $(Q, \Sigma, \Gamma, \delta, q_0, F)$, where $\delta : Q \times \Sigma \rightarrow 2^{Q \times \Gamma_\epsilon}$. After observing a symbol $\sigma \in \Sigma$, the NFT in state $q$ transitions to a choice of $q'$ with output $\gamma \in \Gamma_\epsilon$ where $(q',~\gamma)$ is one of the  pairs in $\delta(q,~\sigma)$. 
\end{definition}

\subsection{Properties, monitors and related terminology}

\cref{sec:motivation} gave two examples of safetly properties modelled as DFAs. More precisely, we model properties using a minimum-state DFA with a special specified error state. We give related definitions here.

\begin{definition}[Finite-state property]
$\phi$ is a finite-state property if it is the minimum-state DFA $\phi = (Q, \Sigma, \delta, q_0, Q \setminus q_{err})$ with the specified error state $q_{err}$. The error state $q_{err}$ must be a \emph{trap state}, i.e. $\forall \sigma \in \Sigma,~\delta(q_{err},~\sigma) = q_{err}$. The notation $Q^\phi, \Sigma^\phi, \delta^\phi, q_0^\phi, q_{err}^\phi$ is used  to refer to $Q, \Sigma, \delta, q_0, q_{err}$ respectively for a property $\phi$. An execution $x \in \Sigma^*$ \emph{violates} the property \(\phi\) if $\delta(q_0, x) = q_{err}$. An execution $x$ that does not violate the property is \emph{non-violating}.
\end{definition}

\begin{remark}
$L(\phi)^C = \Sigma^* \setminus L(\phi)$ are all the strings that violate the property $\phi$.
If an execution violates a property, then so do all its continuations (because $q_{err}$ is a trap state).
\end{remark}

\begin{definition}[NFA property]
A NFA $\psi = (Q, \Sigma, \delta, q_0, Q \setminus q_{err})$ with the specified error state $q_{err}$ is an NFA property. The error state $q_{err}$ must be a \emph{trap state}. If $\psi$ is determinized to a minimum-state DFA $\phi$, then $\phi$ is a finite-state property with the error state $\set{q_{err}}$.
\end{definition}

For a given property $\phi$, a \emph{monitor} $\mon{\phi}$ is synthesized to observe the events that a program generates. The monitor keeps track of the \emph{current state} $q_{curr}$, which is initialized as $q_{curr} = q_0$ and is updated as $q_{curr} \gets \delta(q, \sigma)$ when the symbol $\sigma \in \Sigma$ is observed. A monitor $\mon{\phi}$ produces a \emph{true} verdict -- indicating that the property $\phi$ cannot be violated in any continuation of the observed execution, or a \emph{false} verdict -- indicating that the property has been violated. Till either the \emph{true} or \emph{false} verdict is reached, the verdict is \emph{inconclusive}. If there is a continuation of an execution which leads to the \emph{false} verdict, then the monitor's current state is \emph{monitorable}. In a finite automaton, monitorability of a state $q$ can be checked by checking existence of a path from $q$ to the error state. We use the terms ``monitor" and ``property" interchangeably (e.g. language of a monitor) when it is clear from the context.

\begin{remark}
For the analysis in the following sections, the existence of multiple monitors does not concern us. Therefore we omit any discussion of it. We discuss it when discussing a particular loss type in \cref{sec:more-examples}.
\end{remark}

%% file: incremental-filters.tex
\section{Losses, Alternate Monitors, and Superposed Monitors}
\label{sec:filters}

\input{tikzpictures/formal-definition-count-and-silent.tex}

\newcolumntype{M}[1]{>{\centering\arraybackslash}m{#1}}

We presented two loss models in \cref{sec:motivation}. In this section, we begin with related definitions and describe a class of monitors which can observe lossy streams. We discuss soundness and completeness of these monitors. We then construct optimal monitors to observe lossy streams for a given loss model, and discuss optimality of our construction.

We introduced loss models as a mapping between event symbols or sequence of symbols to alternate symbols. We thus represent a loss model as a relation.%

\begin{definition}[Loss Model]
Let $\Sigma$ and $\Gamma$ be finite alphabets. A \emph{loss model} is defined as a relation $R \subseteq \Sigma^* \times \Gamma$.
\end{definition}

A loss model gives the information about how a single alternate symbol may have been produced. If a symbol $\gamma$ is observed in the lossy stream, then it was produced in lieu of one of the sequence of symbols in $R^{-1}(\gamma)$. 

Consider the lossy stream $2 n 2$ from \ref{fig:safeiter-table} for the corresponding original stream $c n n u u$ ($O_2)$. The program is monitored incrementally, so as it runs, we first observe $2$ in lieu of $c n$, then $n$, and then $2$ in lieu of $u u$. For our theoretical analysis, we wish to address the entire history of how a lossy trace would have been observed, we do that by defining a partial function $f$ on all executions of the event stream, such that $f$ evaluates to the corresponding lossy execution. 

\begin{definition}[Filter and Lossy Streams]

Let $\Sigma$ and $\Gamma$ be finite alphabets. Consider a \emph{loss model} $R \subseteq \Sigma^* \times \Gamma$ . Then a partial function \(f: \Sigma^* \rightarrow \Gamma^*\) defined on all the prefixes of a trace is called a filter under $R$ if it satisfies the \emph{monotonicity} property, defined below:

if $f(x) = y$ and $f(x') \neq y$ for all proper prefixes $x'$ of $x$, then: 
$$f(x \cdot s) = \begin{cases} y \cdot \gamma  & \text{if } s R \gamma\\
y & \text{otherwise}
                               \end{cases}$$
In the first case $\gamma$ is called a \emph{replacement} for the \emph{segment} $s$ of the string $x \cdot s$.
                              
$\FR{R}$ is defined as the set of all possible functions which are filters under $R$.
If $\exists f \in \FR{R}$ such that $f(x) = y \land \nexists x' \prec x, f(x') = y$, then we call $x$ a \emph{completion} for $y$. $x$ is one of the possible executions which could have produced the lossy stream $y$. We define  $\FRinv{R}(y)$ as the set of all completions of $y$:
$$\FRinv{R}(y) \defeq \set{x \in \Sigma^* \mid \exists f \in \FR{R} \sst f(x) \allowbreak = \allowbreak y \allowbreak \land \nexists x' \prec x \land f(x') = y}$$
\end{definition}

A \emph{loss type} is a parameterization over a family of related loss models. Loss type for the loss model from \cref{fig:leading-safeiter-example} is given in \cref{fig:formal-definition-count-and-silent}.

In the next theorem, we see how we can determine the set of completions $\FRinv{R}(y)$ using just $R^{-1}$.

\begin{theorem}
\label{theorem:structure-of-filter-inverse}
For a string $y \in \Gamma^*$, $y = \gamma_1\gamma_2\ldots \gamma_k ~(\forall i~\gamma_i \in \Gamma)$:
$$\FRinv{R}(y) = R^{-1}(\gamma_1)\ldots R^{-1}(\gamma_k)$$
\end{theorem}
\begin{proof}

(LHS $\subseteq$ RHS) Let $x \in \FRinv{R}(y)$, then $\exists$ a partition $x = x_1 \ldots x_k$ such that $\gamma_1,\ldots, \gamma_{k}$ are a replacements for respective $x_i$. Then $x_i \in R^{-1}(\gamma_i)$ and thus $x \in R^{-1}(\gamma_1)\ldots R^{-1}(\gamma_{k})$.

(RHS $\subseteq$ LHS) Let $x \in R^{-1}(\gamma_1)\ldots R^{-1}(\gamma_k) \implies x = x_1 \ldots x_{k} \implies f(x) = \gamma_1 \ldots \gamma_k = y$.\qed
\end{proof}

We now start discussing monitors which observe lossy event streams.

\begin{definition}[Alternate monitor]
Given a \emph{primary} monitor $M_\phi$ and a loss model $R \subseteq \Sigma^* \times \Gamma$, an alternate monitor $M_\psi$ is any finite state monitor over the alphabet $\Gamma$ that observes the lossy execution $f(e)$ for any $f \in \mathcal{F}_R$ when $M_\phi$ observes the execution $e$. We call $(M_\phi, M_\psi)_R$ a \emph{primary-alternate monitor pair} and $(\phi, \psi)_R$ a \emph{primary-alternate property pair}.
\end{definition}

\begin{definition}[Soundness and Completeness for a primary-alternate property pair]
For a primary-alternate pair $(\phi, \psi)_R$, with the definition of $\FRinv{R}$ lifted to the set of strings, we define: \\[4pt]
\emph{Soundness:} A non-violating lossy stream must not have any violating completions, i.e. $y \in L(\psi) \implies \FRinv{R}(y) \subseteq L(\phi)$, equivalently $\FRinv{R}(L(\psi)) \subseteq L(\phi)$  \\[4pt]
\emph{Completeness:} A violating lossy stream must have all violating completions, i.e. $y \not\in L(\psi) \implies \FRinv{R}(y) \subseteq L(\phi)^C$, equivalently $\FRinv{R}(L(\phi)^C) \subseteq L(\phi)^C$  %
\end{definition}

\begin{definition}[Optimality for a primary-alternate property pair]
For a primary-alternate pair $(\phi, \psi^*)_R$ where $\psi^*$ is complete, $\psi^*$ is called optimal if for any other primary-alternate pair $(\phi, \psi)_R$, $L(\psi^*) \subseteq L(\psi)$, or equivalently $\FRinv{R}(L(\psi^*)) \subseteq \FRinv{R}(L(\psi))$
\end{definition}

\begin{remark} Our definition of Optimality is a strong definition. An alternate definition for an optimal monitor might be to count the number of strings up to any given length and define a monitor which reports a violation on maximum number of strings for every length as the optimal monitor, but optimality by our definition would imply optimality in this alternate definition. \end{remark}

It is useful to consider the primary-alternate pair $(M_\phi, M_\psi)_R$ as monitoring together for the purposes of theoretical analysis and for definitions. In practice, we want to monitor using just $M_\psi$.

So far we have only defined alternate monitors, but we have not revealed a strategy to construct them. Our strategy is to keep track of the set of states we could possibly be in. We will now define a special class of alternate monitors to do this.

\begin{definition}[Superposed alternate monitors]
Let $(M_\phi, M_\psi)_R$ be a \linebreak primary-alternate pair where $\phi = (Q, \Sigma, \delta^\phi, q_0, F), F = Q \setminus q_{err}$. $M_\psi$ is called a \emph{superposed} alternate monitor if  $\psi$ is the unique minimum-state DFA for the NFA property $\psi_N = (Q, \Gamma, \delta^\psi, q_0, F)$ where the transition function $\delta^\psi$ satisfies the \emph{superposed monitor condition}, given as follows:

Let $\psi$'s states be labelled by the subsets of $Q$ (this labelling is well-defined, see \cref{lemma:well-defined-NFA-labels} and \cref{remark:well-defined-powerset-minimized-labels} below).
When $M_\psi$ transitions to a state $S \subseteq Q$ and $M_\phi$ is in state $q$, then \(q \in S\). In other words, if $x \in \Sigma^*$ ends with a segment then $\delta^\phi(q_0, x) \in \delta^\psi(q_0, f(x))$ for any filter $f$ under $R$. A superposed monitor is in an \emph{imprecise} state if for its state \(S\), \(|S| \geq 2\).
\end{definition}

A NFA $A = (Q, \Sigma, \delta, q_0, F)$ is \emph{determinized} (converted to a DFA) $B = (2^Q,\Sigma,\delta,~q_0,F')$ where $F' = \set{S \subseteq Q \mid S \cap F \neq \phi} \sst L(A) = L(B)$ \cite{sipser13}. The determinized DFA $B$'s states are labelled by subsets of $Q$. This gives us a relationship between the states of the NFA and its corresponding DFA. Now we show that if we  further consider the minimum-state DFA for $L(B)$ (hence $L(A)$), we can still label its states by subsets of $Q$.

\begin{lemma}
\label{lemma:well-defined-NFA-labels}
In DFA minimization \cite{sipser13} of a determinized NFA, let $\mathcal{P}$ be the paritition of $2^Q$ where $\mathcal{S} \in \mathcal{P}$ represents a set of states merged together. If the states $S_1$ and $S_2$ merge, then the state $S_1 \cup S_2$ merges with them. i.e. $\forall~\mathcal{S} \in \mathcal{P}, \forall~S_1, S_2 \in \mathcal{S} \implies S_1 \cup S_2 \in \mathcal{S}$.\proofinappendix
\end{lemma}

\begin{remark}
\label{remark:well-defined-powerset-minimized-labels}
Using the previous lemma, for each class $[S]_\mathcal{P}$ of states, the class representative $rep^\mathcal{P}(S)$ of $S$ is defined as $\cup_{T \in [S]_P}T$. We label the resultant state from merged states in $[S]_P$ in the minimized DFA by $rep^{P}(S)$. %
\end{remark}

The intuition behind the superposed monitor condition is that we are always over-approximating the state (by a set of states) which we would have been in when monitoring without losses. Due to this over-approximating nature, superposed monitors are always complete, as proved next.

\begin{theorem}
\label{theorem:all-superposed-complete}
All superposed monitors are complete.\end{theorem}
\begin{proof} We give a direct proof. For a superposed monitor $M_\psi$ in state $q$:
\begin{align*}
 y \not\in L(M_\psi) \implies & \delta'(q_0, y) = \set{q_{err}} 
   \implies q \in \set{q_{err}} \implies q = q_{err} 
\\ \implies & \forall x \in \FRinv{R}(y), \ \delta(q, x) = \set{q_{err}} \implies \FRinv{R}(y) \subseteq L(M)^C\tag*{\qed}
\end{align*}
\end{proof}

\begin{definition}[$L_{opt}(\phi, R)$]
For a property $\phi$ and loss model $R$, $L_{opt}(\phi, R) \defeq \FR{R}(L(\phi))$, i.e. $L_{opt}$ is the set of lossy strings in $\Gamma^*$ produced by a non-error execution in $\Sigma^*$. $L_{opt}$ is the smallest set of strings on which a complete alternate monitor cannot reach a false verdict. i.e. $L_{opt} = \FR{R}(L(\phi)) = \set{y \mid \FRinv{R}(y) \cap L(\phi) \neq \varnothing}$.
\end{definition}

Our next theorem gives the construction of an optimal monitor, along with proof of optimality.
\begin{theorem}
\label{theorem:optimality}
Given a property $\phi$ and a loss model $R$, we construct the NFA property $\psi^*$ as the superposed monitor whose transition function $\delta^\psi$ is defined as $\forall q \in Q,~\delta^\psi(q, y) = \delta(q, R^{-1}(y))$. $\psi^*$ recognizes $L_{opt}$.
\end{theorem}
\begin{proof}
\textit{Subproof 1: $y \not\in L(\psi^*) \implies y \not\in L_{opt}$.} This is the same as the completeness criterion and is implied by \cref{theorem:all-superposed-complete}.

\textit{Subproof 2: $y \in L(\psi^*) \implies y \in L_{opt}$.} Consider $y \in L(\psi^*)$.
\begin{align*}
\implies & \delta^\psi(\set{q_0}, y) \neq \set{q_{err}} \\
\implies & \delta^\psi(\ldots\delta^\psi(\delta^\psi(\set{q_0}, y_1), y_2)\ldots), y_n) \neq \set{q_{err}} \\
\implies & \delta^\phi(\ldots\delta^\phi(\delta^\phi(\set{q_0}, R^{-1}(y_1)), R^{-1}(y_2))\ldots), R^{-1}(y_n)) \neq \set{q_{err}} \\
\implies & \delta^\phi(\set{q_0}, \FRinv{R}(y)) \neq \set{q_{err}} \implies y \in L_{opt} & \tag*{\qed}
\end{align*}
\end{proof}

\begin{corollary}
A property $\phi$ is monitorable under loss model $R$ iff state $\tset{q_{err}}$ is reachable in $\psi^*$.
\end{corollary}

\begin{remark} \label{remark:polytime-rinv} Because $R$ can be arbitrary, this construction is only valid if $R^{-1}(y)$ and $\delta^\phi(q, R^{-1}(y))$ are computable. If $R$ is representable by a NFT, then both of these are polynomial time computable. \proofinappendix \end{remark}

\begin{figure}[t]
    \centering
    \begin{subfigure}[m]{0.38\textwidth}
	$\delta^{\psi*}(q, x) $\begin{flushright} = $ \begin{cases}  \delta(q, \Sigma^x) & if\ x \in \tset{1 \ldots n}\\  \delta(q, x) & otherwise \end{cases}$ \end{flushright}
	\caption{Disabling monitoring for up to n events, defined on filter from \cref{fig:dropped-count-filter}}
	\label{example:disabling-monitoring-for-up-to-n-events}
	\end{subfigure}
    \hfill
	\begin{subfigure}[m]{0.60\textwidth}
	\centerline{$\delta^{\psi*}(S, x) = C_\Delta(\delta(C_\Delta(S), x))$}
	$C_\Delta$ is the $\Delta$-closure of the set of states $S$ in $M$, i.e. set of all states which can be reached from states in $S$ by following 0 or more $y$-transitions, where $y \in \Delta$
	\caption{Silent drop monitor, defined on $\Delta$ and the filter from \cref{example:silent-drop-filter}}
	\label{fig:silent-drop-monitor}
	\end{subfigure}
	\caption{Example constructions of $\delta^{\psi^*}$ for $(\phi, \psi^*)_R$ as in \cref{theorem:optimality}}
	\label{fig:example-monitors}
\end{figure}

\cref{fig:example-monitors} and the next section show example optimal monitor constructions.

We have given a liberal definition for a loss model, it can be an arbitrary relation between $\Sigma^*$ and $\Gamma$. We now show that all loss models of interest can actually be represented by a more restricted definition -- a relation which must be representable by a NFT.

\begin{theorem}
\label{theorem:rational-filter}
Let $(\phi, \psi^*)_R$ be the primary-alternate property pair as constructed in \cref{theorem:optimality} where $R$ may not be representable by NFT. Then there exists a loss model $R'$ which can be represented as a NFT for which the constructed alternate property is also $\psi^*$.
\end{theorem}

The following definitions are required for the proof of \cref{theorem:rational-filter}.

\begin{definition}[Generalized Nondeterministic Finite Automaton \cite{Han04}]
A generalized nondeterministic finite automaton \emph{(GNFA)} is a 5-tuple $(Q, \Sigma, \delta, q_0, f)$, where $Q$ is the finite set of states, $\Sigma$ is the alphabet, $\delta \subseteq (Q \setminus f) \times (Q \setminus q_0) \rightarrow REG$ is the transition function, and $q_0, f \in Q$ are the specified initial and final states.
\end{definition}
\begin{remark}
A GNFA can be converted to a NFA \cite{sipser13}.
\end{remark}

\begin{definition}[Generalized Nondeterministic Finite-State Transducers (GNFTs)]
Defined as a GNFA $(Q, \Sigma, \Gamma, \delta, q_0, f)$, where $\delta : (Q \setminus f) \times (Q \setminus q_0) \rightarrow 2^{REG(\Sigma) \times \Gamma}$. After observing a string $x \in \Sigma^*$, the NFT in state $q$ transitions to a choice of $q'$ with output $\gamma \in \Gamma$ where $(r,~\gamma)~s.t.~x \in L(r)$ is one of the pairs in $\delta(q,~q')$.
\end{definition}
\begin{remark}
A GNFT can be converted to a NFT. A sketch of this conversion which works by expanding each transition follows. Take a transition $(r,~\gamma) \in \delta(q_1,~q_2)$. The regex $r$ can be converted to a NFA $A = (Q', \Sigma, \delta', q_0', f')$ with a single accept state by standard algorithms. Now this NFA can be embedded in the GNFT in place of the transition $(r,~\gamma) \in \delta(q_1,~q_2)$ by merging the state $q_1$ with $q_0'$ and adding the output $\gamma$ to all transitions into $f'$ and merging $f'$ with $q_2$. Finally, remove the transition $(r,~\gamma)$ from $\delta(q_1,~q_2)$.
\end{remark}

\begin{proof}[\cref{theorem:rational-filter}]
Let $\delta$ be the tranistion function for $\phi$. The construction in \cref{theorem:optimality} uses $\delta(q,~R^{-1}(\gamma))$ for defining $\psi^*$'s transition function. Therefore it is sufficient to produce an $R'$ representable by a NFT such that $\delta(q,~R'^{-1}(\gamma)) = \delta(q,~R^{-1}(\gamma))~\forall\gamma \in \Gamma$. 

Consider a symbol $\gamma \in \Gamma$.

\paragraph{\textbf{Case 1:} ($R^{-1}(\gamma)$ is regular)} We define $xR'\gamma~\forall xR\gamma$. So $R'^{-1}(\gamma) = R^{-1}(\gamma)$ and thus  $\delta(q,~R'^{-1}(\gamma)) = \delta(q,~R^{-1}(\gamma))$

\paragraph{\textbf{Case 2:} ($R^{-1}(\gamma) $ is not regular)}

We'll use the shorthand $f(q): Q \rightarrow 2^Q$ for $f_y(q) =_{def} \delta(q, R^{-1}(\gamma))$. 

Consider $l(q): Q \rightarrow REG(\Sigma)$, the regular language taking us from $q$ to $f(q)$, i.e. $l(q) =_{def} \set{ x \mid x \in \Sigma^* \land \delta(q, x) \in f(q)}$.

Let $L_\gamma \defeq \cap_{q \in Q}~l(q)$. It follows that $L_\gamma$ is regular since regular languages are closed under intersection.

We note that $R^{-1}(\gamma) \subseteq l(q)~\forall q \in Q$, and thus $R^{-1}\gamma \subseteq L_\gamma$.

We define $xR'\gamma~\forall x \in L_\gamma$. It is left to prove that $\delta(q, L_\gamma) = f(q)~\forall q \in Q$.

We show $\delta(q, L_y) \subseteq f(q)\ \forall q\ \because L_y \subseteq l(q)$ and $\delta(q, L_y) \supseteq f(q)\ \forall q \in Q\ \because  R^{-1}(y) \subseteq L_y$.

We construct a GNFT for $R'$. Consider a GNFT with states $\set{q_0, q, f}$, $\epsilon$-transitions from $q_0$ to $q$ and $q$ to $f$, and self loop edges on $q$ $\forall \gamma \in \Gamma$ with input label as the regex of $R^{-1}(\gamma)$ and output label $\gamma$. This completes the construction.\qed
\end{proof}

Moreover, we have also shown that an even more relaxed definition of loss model (a relation between $\Sigma^*$ and $\Gamma^*$ (instead of  $\Sigma^*$ and $\Gamma$) will not increase the number of loss models we can express. The two crucial results below complete our claim that \emph{any} loss type (arbitrary relation between original and alternate symbols) for which the final produced property has to be be finite-state, is representable in our framework.

If $R$ is defined over $\Sigma^* \times \Gamma^*$, we can always come up with a $R' \subseteq \Sigma^* \times \Gamma'$ and alter the transitions of the alternate automata such that transitions for same input strings in $\Sigma^*$ retain the same semantics for state transitions as the unaltered automata. We prove this in the next theorem.

\begin{theorem}
Let $N$ be a finite state monitor over a filter $f:\Sigma^* \rightarrow \Gamma^*$ under $R$. Then $\exists$ $R'$ with a finite range such that $\exists f':\Sigma^* \rightarrow \Gamma^*$, a filter over $R'$ such that $ \forall q \in Q\ \forall x \in \Sigma^*\ \delta(q, f(x)) = \delta(q, f'(x))$.
\end{theorem}
\begin{proof} We give a non-constructive proof.
Partition all strings in $\Gamma^*$ using the relation $\sim$ defined as : $x \sim y \iff \forall q \in Q\ \delta(q, x) = \delta(q, y)$

It is easy to see that $\sim$ is reflexive, symmetric and transitive.

There are only $\left|Q\right|^{\left|Q\right|}$ possibilities for $\delta(q, x)$ for a fixed $x$. Thus, the number of equivalence classes is upper-bound by $\left|Q\right|^{\left|Q\right|}$ and is finite.

Choose a class representative for each equivalence class arbitrarily (e.g. the lexicographically least string in that class). Let $[y]$ denote the equivalence class of $y$ and $[y]_r$ denote the class representative.

Define $R' = \set{(x, [y]_r) \mid (x, y) \in R}$

Define $f'(x) = [y_1]_r \cdot [y_2]_r \cdot \ldots  \cdot [y_n]_r$ where $f(x)$ has segments $x_1, \ldots  x_{n+1}$ and replacements $y_1, \ldots , y_n$.

We prove by induction on segments that $ \forall q \in Q \forall x \in \Gamma^* \delta(q, f(x)) = \delta(q, f'(x))$.

\textbf{Base Case:} 0 segments

$f(\varepsilon) = \varepsilon$ and $f'(\varepsilon) = \varepsilon$. $\therefore$ $\forall q \delta(q, f(x)) = \delta(q, f'(x)) = q$

\textbf{Induction Hypothesis:} $\forall q\ \delta(q, f(x)) = \delta(q, f'(x))$ where $x$ has $< k$ segments

\textbf{Induction Step:} Consider $x \in \Sigma^*$ with segments $x_1, \ldots , x_k$.

Let $x' = x_1 \cdot \ldots  \cdot x_{k-1}$

Then $\delta(q, f(x))' = \delta(q, f'(x'))$ by induction hypothesis. Let these both equal $q'$.

Let $f(x) = f(x')y$ and $f'(x) = f(x')y' = f(x')[y]_r$

Since $\forall q'\ [y]_r \sim y$, $\delta(q', [y]_r) = \delta(q', y)$

Thus by induction we have proved that $\delta(q, f(x)) = \delta(q, f'(x)$ for all $q \in Q$ where $x \in \Gamma^*$ with any number of segments, which is just any $x \in \Gamma^*$
\end{proof}

\begin{remark}[Sound alternate monitors]
\label{remark:sound-alternate-monitors}
We can also construct a primary-alternate pair $(\phi, \psi)_R$ which is sound and may be incomplete by using a construction similar to that in \cref{theorem:optimality} by determinizing it and updating $\delta^\psi(S, ~\gamma) \gets \tset{q_{err}}$ if $q_{err} \in \delta(S,~R^{-1}(\gamma))$. It can be argued in a similar fashion that this construction is optimally complete among all sound lossy monitors. \proofinappendix
\end{remark}

%% file: tikzpictures/formal-definition-count-and-silent.tex
\begin{figure}[t]
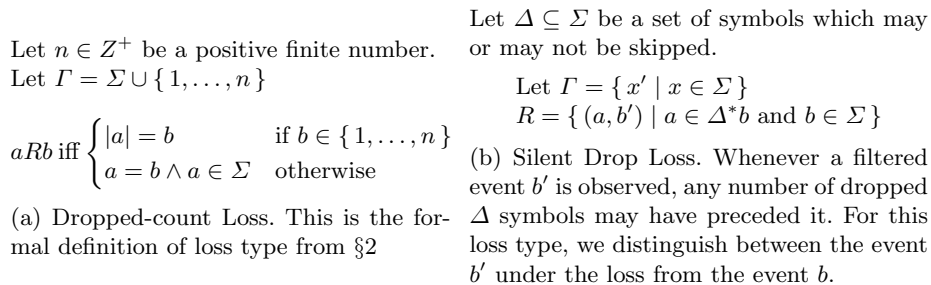

    \centering
    \begin{subfigure}[m]{0.48\textwidth}
    
    Let \(n \in Z^+\) be a positive finite number.\\
    Let \(\Gamma = \Sigma \cup \set{1, \ldots, n}\)\\
    
    $aRb$ iff $\begin{cases}\left| a \right| = b & \text{if } b \in \set{1, \ldots, n} \\ a = b \land a \in \Sigma & \text{otherwise} \end{cases}$
    \caption{Dropped-count Loss. This is the formal definition of loss type from \cref{fig:leading-safeiter-example}}
    \label{fig:dropped-count-filter}
    \end{subfigure}
    \hfill
    \begin{subfigure}[m]{0.50\textwidth}
    
    Let \(\Delta \subseteq \Sigma\) be a set of symbols which may or may not be skipped.\vspace{5pt}
    \begin{center}
    \begin{varwidth}{\textwidth}
    Let $\Gamma = \set{x' \mid x \in \Sigma}$ \\
    $R = \set{(a, b') \mid a \in \Delta^*b \text{ and } b \in \Sigma}$
    \end{varwidth}
    \end{center}
    \caption{Silent Drop Loss. Whenever a filtered event $b'$ is observed, any number of dropped $\Delta$ symbols may have preceded it. For this loss type, we distinguish between the event $b'$ under the loss from the event $b$.}
    \label{example:silent-drop-filter}
    \end{subfigure}
    \caption{Formally specified loss types}
    \label{fig:formal-definition-count-and-silent}
\end{figure}

%% file: more-examples.tex
\section{Framework Instantiations}
\label{sec:more-examples}
We have already shown the applicability of our framework to loss types such as those in \cref{fig:formal-definition-count-and-silent}. In this section, we describe three more instantiations of the framework that illustrate the variety of realistic event loss models it can accommodate.

\paragraph{\bf Counting frequency of missed symbols up to n missed symbols}
\label{example:frequency-count}
This loss model was considered in \cite{Falzon13} for lossily compressing event traces over a slow network. It is a modification of the dropped-count filter in \cref{fig:dropped-count-filter} where additional information about symbols is kept. 

Let \(n \in Z^+\) and let the symbols in $\Sigma$ be indexed by $I = \set{1 \ldots \abs{\Sigma}}$ and be denoted by $\sigma_i$ where $i\in I$. Define $\Gamma$ as $\set{(c_1, c_2, \ldots, c_\abs{\Sigma}) | 0 < c_1 + \ldots + c_{\abs{\Sigma}} \leq n}$.
The loss model is defined as $R = \set{(x,(c_1, c_2, \ldots c_{\abs{\Sigma}})) \mid \bigwedge_{i \in I} c_i = \#\sigma_i(x)}$.

We discuss two key differences between our formalization and that of \cite{Falzon13}. First, the total size of the missed symbols is bounded in our case so that we have a finite alphabet with each transition taking $O(1)$ time in the determinized alternate DFA. \cite{Falzon13} uses a constraint automata which accepts an infinite alphabet and each transition takes $O(\abs{Q})$ time. We note that even if more than $n$ symbols are missed at a time, then up to $mn$ missed symbols can produced $m$ alternate symbols to transition to the correct set of states in our framework. The second difference is in consideration of soundness and completeness. While we construct a complete optimal monitor (without any false positives), they construct a sound monitor (without any false negatives). As \autoref{remark:sound-alternate-monitors} shows, this is not an issue since we can easily construct a sound monitor instead.

\paragraph{\bf Merged Objects}
\label{example:merged-objects-filter}

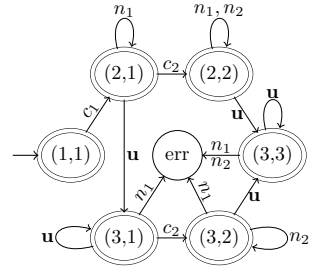
\begin{wrapfigure}[12]{r}{0.35\textwidth}
        \centering
        \vspace*{-12mm}
        \input{tikzpictures/composite-monitor-original.tex}
        \caption{A composite monitor for SafeIter on two iterators \cite{Purandare13}.  $(i, j)$ represents states $(q_i, q_j)$ for two different iterators and the subscript $1$ on events is assigned to the monitor created first.}
        \label{fig:safeiter_composite_monitor}
\end{wrapfigure}

Here we look at a new loss type which loses information about which object an event belongs to in a multi-object monitor.
Let $O = \set{o_1 \ldots o_m}$ be a set of objects with parametric events $E = e_1 \dots e_n$, i.e. $e_1(o_1)$ is a distinct event from $e_1(o_2)$. This means that $\Sigma = \set{e(o) \mid e \in E \land o \in O}$.

For $\sigma \in \Sigma, \gamma \in \Gamma$, let $\sigma R\gamma $ iff $\sigma = \gamma(o) \land o \in O$. In the filtered event stream, we lose information about which object the event belongs to within $O$.
For the general case we can build an optimal monitor using the construction in \cref{theorem:optimality}. We give an example of the  multi-object property ``SafeIter" shown in \cref{fig:safeiter-property} which states that a collection object should not be updated while an iterator object on that collection iterates. 
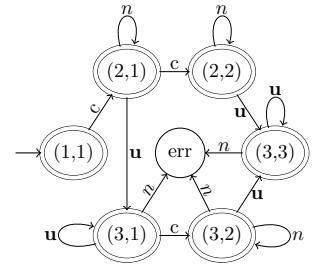
\begin{wrapfigure}[10]{r}{0.35\textwidth}
        \centering
        \vspace*{-13mm}
        \input{tikzpictures/composite-monitor-transformed.tex}
        \caption{Optimal complete alternate monitor.}
        \label{fig:safeiter_lossy}
\end{wrapfigure}

A composite monitor for the SafeIter property from \cite{Purandare13} is shown in \cref{fig:safeiter_composite_monitor} for two iterators $I_1$ and $I_2$. The loss model $R$  merges events for $I_1$ and $I_2$, and using \autoref{theorem:optimality} to construct the optimal alternate monitor, we obtain the monitor in \cref{fig:safeiter_lossy}. 
For example, the obtained monitor only misses the violation for the event streams like $c_1 u c_2 n_2\text{*} n_1$, i.e. when we actually need the information that event $n$ happened on object $1$, but can still report violations for event streams matching $c_1 n_1 c_2 (n_1 | n_2) uu\text{*}(n_1 | n_2)$ or $c_1 u c_2 (n_2 | n_1)\text{*} uu\text{*} (n_1 | n_2)$. 

\paragraph{\bf Missing Loop Events}
\input{tikzpictures/stuttering-figure.tex}
Significant number of events can be generated within loops in a program. \cite{Purandare10} addresses this by eliminating instrumentation losslessly within loops when monitoring the first few iterations is sufficient.%

We consider an extension of this idea in \cref{fig:stuttering} where the program structure is used to obtain the loss model. Instrumentation from the loop is 
replaced with a single symbol ``k" at the end of the loop. If instrumentation is disabled for all iterations of the loop, the monitor is in states $\set{q_2, q_3}$ after the event ``k". If the first few iterations are monitored and event ``a" is generated, the monitor will be in states $\set{q_2}$ after the event ``k".

The loss model can be calculated using a method from \cite{Dwyer07}. It presents a static analysis which finds the set of states that are possible after a program region, e.g., a loop body, for any given starting state if monitoring were to be disabled in that region. We can use this information directly instead of $R^{-1}(k)$ in \autoref{theorem:optimality} for computing $\delta'(q,~k)~\forall q$. This is equivalent to mapping the set of strings which go from $q$ to $\delta'(q,~k)$ to the new symbol for the loss model.

%% file: tikzpictures/composite-monitor-original.tex
\usetikzlibrary{positioning,automata,calc,shapes.geometric,arrows,fit} 
\begin{tikzpicture}[inner sep=1,node distance=0.4,scale=0.75,every node/.style={scale=0.75}] %

  \tikzset{every state/.style={shape=ellipse,very thin}}
  
  \node[state,accepting,initial]   (q_{11})                    {(1,1)}; 
  \node[state,accepting]           (q_{21})  [above right= 0.6 and 0.1 of q_{11}] {(2,1)}; %
  \node[state,accepting]           (q_{22})  [right=of q_{21}] {(2,2)};
  \node[state,accepting]           (q_{33})  [below right= 0.6 and 0.1 of q_{22}] {(3,3)};
  
  \node[state,accepting]           (q_{31})  [below = 1.5 of q_{21}] {(3,1)};
  \node[state,accepting]           (q_{32})  [right= of q_{31}] {(3,2)};    

  \node[state]                     (q_{err}) [left = 0.5 of q_{33}] {err};
  
  \draw[->,very thin] 
            (q_{11})     edge                 node [above,sloped] {$c_1$} (q_{21})
            (q_{21})     edge [loop above]    node         {$n_1$}    ()
            (q_{21})     edge                 node [above] {$c_2$} (q_{22})
            (q_{22})     edge [loop above]    node         {$n_1, n_2$}   ()
            (q_{22})     edge                 node [left] {\bf u} (q_{33})
            (q_{33})     edge [loop above]    node         {\bf u} ()
            (q_{33})     edge                 node [above,sloped] {$n_1 $} node [below,sloped] {$n_2$}   (q_{err})
            
            (q_{21})     edge                 node [right]  {\bf u}    (q_{31})
            (q_{31})     edge                 node [above]        {$c_2$} (q_{32})
            (q_{31})     edge [loop left]     node                {\bf u}    (q_{31})
            (q_{31})     edge                 node [above,sloped] {$n_1$} (q_{err})
            (q_{32})     edge [loop right]    node                {$n_2$} ()
            (q_{32})     edge                 node [above,sloped] {$n_1$} (q_{err})
            (q_{32})     edge                 node [right] {\bf u}    (q_{33});
\end{tikzpicture}

%% file: tikzpictures/composite-monitor-transformed.tex
\usetikzlibrary{positioning,automata,calc,shapes.geometric,arrows,fit}%
\begin{tikzpicture}[inner sep=1,node distance=0.4,scale=0.75,every node/.style={scale=0.75}] %
  \tikzset{every state/.style={shape=ellipse,very thin}}
  \node[state,accepting,initial]   (q_{11})                    {(1,1)}; 
  \node[state,accepting]           (q_{21})  [above right= 0.6 and 0.1 of q_{11}] {(2,1)}; %
  \node[state,accepting]           (q_{22})  [right=of q_{21}] {(2,2)};
  \node[state,accepting]           (q_{33})  [below right= 0.6 and 0.1 of q_{22}] {(3,3)};
  
  \node[state,accepting]           (q_{31})  [below = 1.5 of q_{21}] {(3,1)};
  \node[state,accepting]           (q_{32})  [right= of q_{31}] {(3,2)};    

  \node[state]                     (q_{err}) [left = 0.5 of q_{33}] {err};

  \draw[->,very thin] 
            (q_{11})     edge                 node [above,sloped] {c}     (q_{21})
            (q_{21})     edge [loop above]    node                {$n$}     ()
            (q_{21})     edge                 node [above]        {c}     (q_{22})
            (q_{22})     edge [loop above]    node                {$n$}     ()
            (q_{22})     edge                 node [left] {\bf u} (q_{33})
            (q_{33})     edge [loop above]    node                {\bf u} ()
            (q_{33})     edge                 node [above,sloped] {$n$}     (q_{err})
            
            (q_{21})     edge                 node [above,right]  {\bf u} (q_{31})
            (q_{31})     edge                 node [above]        {c}     (q_{32})
            (q_{31})     edge [loop left]     node                {\bf u} (q_{31})
            (q_{31})     edge                 node [above,sloped] {$n$}     (q_{err})
            (q_{32})     edge [loop right]    node                {$n$}     (q_{err})
            (q_{32})     edge                 node [above,sloped] {$n$}     (q_{err})
            (q_{32})     edge                 node [right] {\bf u} (q_{33});
\end{tikzpicture}

%% file: tikzpictures/stuttering-figure.tex
\begin{figure}[t]
    \begin{subfigure}[m]{0.22\textwidth}
        \input{tikzpictures/stuttering-program.tex}
        \caption{Loop targeted for instrumentation removal.}
    \end{subfigure}
    \hfill
    \begin{subfigure}[t]{0.23\textwidth}
        \input{tikzpictures/stuttering-original.tex}
        \caption{Program property DFA. All missing edges go to error state (not shown).}
    \end{subfigure}
    \hfill
    \begin{subfigure}[t]{0.22\textwidth}
        \input{tikzpictures/stuttering-transformed.tex}
        \caption{Alternate monitor NFA. Error state not shown.}
    \end{subfigure}
    \hfill
    \begin{subfigure}[t]{0.28\textwidth}
        \input{tikzpictures/stuttering-filter.tex}
        \caption{A filter to remove instrumentation from loop and replace it with a single symbol ``k".}
    \end{subfigure}
    \caption{Missing Events in Loops to be able to remove instrumentation in them. } 
    \label{fig:stuttering}
\end{figure}
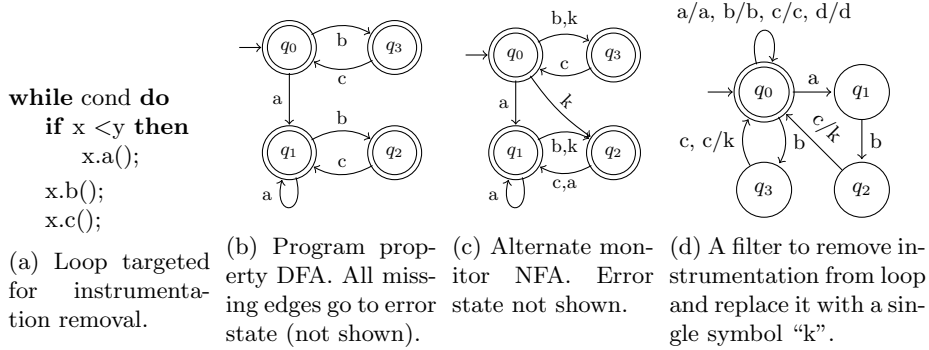

%% file: tikzpictures/stuttering-program.tex
\newenvironment{algo}
 {\par\addvspace{\topsep}
  \centering
  \begin{minipage}{\linewidth}}
 {\end{minipage}
  \par\addvspace{\topsep}} %
\makeatletter %
\xpatchcmd{\algorithmic}
  {\ALG@tlm\z@}{\leftmargin\z@\ALG@tlm\z@}
  {}{}
\makeatother
\begin{algo}
\begin{algorithmic}[l]
  \While{cond} 
    \If{x \textless y}
        \State x.a();
    \EndIf
    \State x.b();
    \State x.c();
  \EndWhile
\end{algorithmic}
\end{algo} 

%% file: tikzpictures/stuttering-original.tex
\usetikzlibrary{positioning,automata,calc,shapes.geometric,arrows,fit} 
\begin{tikzpicture}[shorten >=1pt,node distance=1.4,on grid, scale=0.8]
  \tikzset{every node/.style={scale=0.8}}
  \node[state,accepting,initial]    (q_0)                             {$q_0$}; 
  \node[state,accepting]            (q_1)    [below=of q_0]  {$q_1$}; 
  \node[state,accepting]            (q_2)    [right=of q_1]           {$q_2$}; 
  \node[state,accepting]            (q_3)    [right=of q_0]  {$q_3$}; 
  
  \path[->] (q_0)     edge                node [left]        {a}(q_1)

            (q_1)     edge [loop below]   node [above left= 0 and 4pt]        {a} (q_1)
            (q_1)     edge [bend left]    node [above]        {b} (q_2)
            (q_2)     edge [bend left]    node [above]        {c} (q_1)

            (q_0)     edge [bend left]    node [below]        {b} (q_3)
            (q_3)     edge  [bend left]   node [below]        {c} (q_0)
            ;
\end{tikzpicture}

%% file: tikzpictures/stuttering-transformed.tex
\usetikzlibrary{positioning,automata,calc,shapes.geometric,arrows,fit} 
\begin{tikzpicture}[shorten >=1pt,node distance=1.3cm,on grid, scale=0.8]
  \tikzset{every node/.style={scale=0.8}}
  \node[state,accepting,initial]    (q_0)                             {$q_0$}; 
  \node[state,accepting]            (q_1)    [below=of q_0]  {$q_1$}; 
  \node[state,accepting]            (q_2)    [right=of q_1]           {$q_2$}; 
  \node[state,accepting]            (q_3)    [right=of q_0]  {$q_3$}; 
    
  \path[->] (q_0)     edge                node [left] {a} (q_1)
            (q_0)     edge [bend right=10]    node [above,sloped]  {k} (q_2)

            (q_1)     edge [loop below]   node [above left= 0 and 4pt]  {a} (q_1)
            (q_1)     edge [bend left]    node [below]        {b,k} (q_2)
            (q_2)     edge [bend left]    node [below]        {c,a} (q_1)

            (q_0)     edge [bend left]    node [above]        {b,k} (q_3)
            (q_3)     edge  [bend left]   node [above]        {c} (q_0)
            ;
\end{tikzpicture}

%% file: tikzpictures/stuttering-filter.tex
\usetikzlibrary{positioning,automata,calc,shapes.geometric,arrows,fit} 
\begin{tikzpicture}[shorten >=1pt,node distance=1.3,on grid, scale=0.9]
  \tikzset{every node/.style={scale=0.9}}
  \node[state,accepting,initial]    (q_0)                          {$q_0$}; 
  \node[state]                      (q_1)     [right=of q_0]       {$q_1$}; 
  \node[state]                      (q_2)     [below=of q_1]       {$q_2$}; 
  \node[state]                      (q_3)     [below=of q_0]       {$q_3$}; 
  
  \path[->] (q_0)     edge [loop above]   node [above]      {a/a, b/b, c/c, d/d} (q_0)
            (q_0)     edge                node [above]      {a}    (q_1)
            (q_1)     edge                node [right]      {b}              (q_2)
            (q_2)     edge                node [above,sloped]      {c/k}            (q_0)

            (q_0)     edge [bend left]    node [right]      {b}              (q_3)
            (q_3)     edge [bend left]    node [left]       {c, c/k}         (q_0)
            ;
\end{tikzpicture}

%% file: emperical-study.tex
\section{Empirical Study}
\label{sec:evaluation}

\input{tikzpictures/eval-table-paper.tex}

We implemented the optimal monitor construction algorithm and dropped-count loss type \cref{fig:dropped-count-filter} to qualitatively analyze the behavior of optimal monitors under losses. The main goal of the study is to see if the optimal monitor is \emph{effective} in detecting violations over lossy event streams. 
While the primary contribution of this work is theoretical -- the optimality of our construction has been proven -- it is still informative
to assess the potential for practical impact.
As with similarly oriented work, e.g., \cite{Gange17ATVA}, we use a simulation study for
this purpose and leave the engineering of efficient tooling to future work.

We address different aspects of effectiveness by exploring the following research questions.

\textbf{RQ1:} How many violations did the optimal alternate monitor miss? 

A trivial monitor can miss all violations and still be ``complete", since it produces no false positives. A lossless monitor misses no violations. Our optimal alternate monitor lies somewhere between the two, and we wish to measure where.

\begin{figure}[t]
\begin{subfigure}[5]{0.52\textwidth}
	\hspace{-8pt}\includegraphics[scale=0.45]{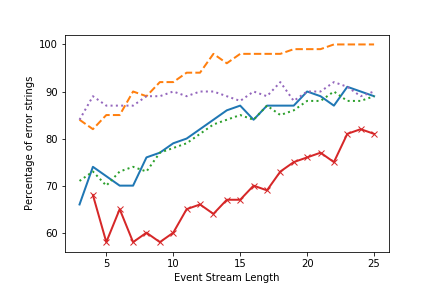}
	\caption{String length vs percentage of strings a violation was observed on. These 5 properties show trends which are representative of all properties. %
	}
	\label{fig:chart}
\end{subfigure}
\hfill
\begin{subfigure}[r]{0.44\textwidth}
    \vspace{2mm}
    \hspace{-21pt}\includegraphics[scale=0.55]{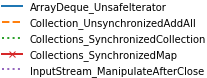}%
    \vspace{-5mm}
    \input{tikzpictures/disable-probability-and-length-events-processed.tex}
    \vspace{-5mm}
    \caption{Average percentage of violating strings detected across settings of $\rho$ and $\eta$.}
	\label{fig:rho-eta-table}
	\end{subfigure}
\caption{Summarized results. A complete version of (a) showing all 26 properties is available in the appendix.}
\label{fig:summarized-results}
\end{figure}
\textbf{RQ2:} How many events did the optimal alternate monitor have to process?

The main motivation behind induced event loss is overhead control. We explore the relation between the events missed and impact on alternate monitor's ability to report an error.

\subsection{Methodology}

We selected a number of properties from \cite{bricsautomaton}. We mined 157 property specification files from \texttt{runtimeverification.com}, and collected a subset of 26 properties after de-duplication for which the specification contains a regular expression describing the property. We also filtered out properties which are trivial for dropped count loss model (e.g. properties that require an event occurs at most once).

Many of these properties specify one or more \emph{creation} events that are used to instantiate new monitors according to the event's context~\cite{javamop}. Skipping these events makes monitoring impossible for subsequent events related to these monitors. 

We took special care to ensure that such creation events were injected into the event stream appropriately in our study.

Our implementation reads in the property specification files and extracts 1) all events, 2) creation events, 3) the regular expression describing the property, and 4) \texttt{@match} or \texttt{@fail} keywords which specify if a violation occurs when regex is matched or when it fails to match, respectively. It then creates minimized DFAs with an error state from these regexes with the help of \texttt{brics.automaton} library \cite{bricsautomaton}. 

After this pre-processing, our implementation reads in the description of property DFA $\phi$ and computes a bounded drop loss model (using n = 5 as the bound) on its alphabet $\Sigma$. 
This loss model and the property DFA are then used to create an optimal monitor NFA, which is  determinized and minimized to give an optimal alternate monitor $\psi$.

To explore variation in monitor performance with trace length, we generated traces of length $n$ as follows:
\begin{enumerate}
    \item If the property had a non-empty set $C$ of creation events, the first event was chosen randomly from $C$, and rest $n - 1$ events were chosen from $\Sigma \setminus C$ uniformly at random;
    \item Otherwise, all $n$ events were chosen from $\Sigma$ uniformly at random.
\end{enumerate}
The generated traces were defined for each property using its alphabet, $\Sigma$; we did not reuse traces across properties even if they share the same alphabet.

Each trace was then subjected to the following procedure to inject artificial loss where count symbols are restricted to the range $1, \ldots, 5$.
The procedure takes two parameters $\rho$ (probability of disabling monitoring) and $\eta$ (mean length of number of events to miss):
\begin{enumerate}
    \item Start at the first element of the trace, if there is a creation event, replicate it in the alternate stream and consider the next event. The following steps are repeated until there are no more events to be processed.
    \item Draw a random number $disable \sim bernoulli(\rho)$ which has probability $\rho$ of being 1 and $1 - \rho$ of being 0. 
    \item If $disable = 1$, draw a random number $l \sim exponential(\eta)$, which is a real number with expected value $\eta$. Ignore the next $m = \ceil*{l}$ symbols in the input stream. If $m$ is not divisible by 5, add a symbol for $m$ modulo $5$ and then $\floor*{(m/5)}$ ``$5$" symbols to the alternate stream.
    \item If $disable = 0$, add the current symbol in the original stream to the alternate stream.
\end{enumerate}

We then simulated the original property monitor $\phi$ and alternate monitor $\psi$ on the original and alternate event streams, respectively.
Simulations ran over $M = 1000$ random traces each for lengths between $n = 3$ to $n = 25$, for 4 combinations of $\rho \in \set{0.1, 0.3}$ (low and high probability of disabling monitoring) and $\eta \in \set{3, 6}$ (low and high disable lengths).
Our simulation recorded whether the monitor exclusively reached the error state in which case it reported a violation.

\subsection{Results}

We summarized the results in \cref{table:raw-data} and \cref{fig:summarized-results}. Here we address our RQs.

\textbf{RQ1:} We see that the optimal monitor reports anywhere from 60\% to 94\% violations for low number of losses ($\rho, \eta = 0.1, 0.3$) to 15-69\% for higher number of losses ($\rho, \eta = 0.3, 6$). There is large variation among the properties, and some are more amenable for reporting losses than others. However, these results indicate that the optimal alternate monitors are capable of detecting errors in lossy event streams. As expected and proved earlier, these monitors did not report a single false positive preserving the completeness of the analysis. An interesting unanswered question left to the future work is to see what structural characteristics of these properties cause the variation in monitorability under losses with respect to different loss types.

\textbf{RQ2:} As we see in \cref{fig:rho-eta-table}, average number of events processed are $84.9\%$ for the low-loss case and $54.2\%$ for the high-loss case. There is a clear trend in \cref{table:raw-data} of the violation percentage decreasing with an increase in incurred losses. Still, it is promising to see that despite so many losses, 17 out of 26 properties are able to report 40\% violations or more.

\subsection{Limitations and Threats to Validity}

An inherent limitation of this study is that it is based on simulation and not on real program traces. All symbols are generated with equal probability for our artificial traces. The traces generated by real programs are likely to be biased towards non-violating behavior for most objects. However, the primary goal of our study is to understand the error detection capability of an alternate monitor. We believe that our randomly generated traces exclusively model the erroneous behavior of the violating parts of the program. Therefore, the results indicating the error detection ability of an alternate monitor on such traces reflect its ability to report errors in real violating behaviors. 

Another limitation is the number of events considered in a trace. Even though real programs generate long traces, they often consist of a large number of short \textit{sub-traces} where each one of which belongs to a different monitor. A sub-trace that belongs to one monitor does not interfere with the analysis of another sub-trace. Therefore, we believe that our choice of generating short but monitor-specific traces is justified.
Moreover, as the number of events increases (refer to Figure~\cref{fig:chart}), the ability of optimal alternate monitor to report a violation tends to be higher. This indicates that shorter traces are more challenging for alternate monitors than longer ones.

%% file: tikzpictures/eval-table-paper.tex
\newcolumntype{Y}{>{\centering\arraybackslash}X} %
\begin{table}[t]
\caption{\protect\centering Error detection for short strings (Length 5 to 10). For strings of lengths 10 - 20, data is available in the appendix.}
\vspace{2mm}
\label{table:raw-data}
\scriptsize
\begin{tabularx}{\textwidth}{|l|Y|Y|Y|Y|}
\hline
\textbf{Property}                        & $\rho: 0.1, \eta: 3$     & $\rho: 0.1, \eta: 6$     & $\rho: 0.3, \eta: 3$ &  $\rho: 0.3, \eta: 6$    \\
\hline
ArrayDeque\_UnsafeIterator               & 75\% (4598) & 71\% (4595) & 39\% (4579) & 33\% (4606) \\
Collections\_SynchronizedCollection      & 75\% (4345) & 69\% (4296) & 37\% (4298) & 30\% (4379) \\
Collections\_SynchronizedMap             & 60\% (1824) & 56\% (1798) & 19\% (1848) & 16\% (1814) \\
Collection\_UnsynchronizedAddAll         & 90\% (4942) & 84\% (4935) & 65\% (4942) & 52\% (4931) \\
Console\_CloseReader                     & 85\% (4739) & 81\% (4761) & 58\% (4764) & 46\% (4801) \\
HttpURLConnection\_SetBeforeConnect      & 86\% (4758) & 81\% (4759) & 57\% (4783) & 46\% (4759) \\
InputStream\_MarkAfterClose              & 86\% (4753) & 80\% (4741) & 58\% (4777) & 46\% (4742) \\
Iterator\_RemoveOnce                     & 87\% (4310) & 83\% (4288) & 63\% (4317) & 57\% (4335) \\
ListIterator\_RemoveOnce                 & 83\% (3118) & 80\% (3139) & 57\% (3092) & 53\% (3085) \\
ListIterator\_Set                        & 83\% (3947) & 79\% (4005) & 50\% (3963) & 43\% (3945) \\
List\_UnsynchronizedSubList              & 74\% (4564) & 71\% (4569) & 37\% (4548) & 31\% (4582) \\
Map\_UnsafeIterator                      & 60\% (1870) & 57\% (1872) & 18\% (1803) & 15\% (1844) \\
Math\_ContendedRandom                    & 94\% (4961) & 91\% (4970) & 82\% (4963) & 71\% (4967) \\
NavigableSet\_Modification               & 60\% (1868) & 57\% (1876) & 19\% (1921) & 15\% (1897) \\
PushbackInputStream\_UnreadAheadLimit    & 86\% (4268) & 80\% (4299) & 61\% (4299) & 48\% (4265) \\
Reader\_ReadAheadLimit                   & 87\% (4150) & 80\% (4144) & 62\% (4151) & 52\% (4121) \\
Reader\_UnmarkedReset                    & 89\% (2434) & 90\% (2442) & 69\% (2466) & 71\% (2453) \\
Scanner\_ManipulateAfterClose            & 75\% (4567) & 71\% (4600) & 38\% (4549) & 31\% (4578) \\
ServerSocket\_SetTimeoutBeforeBlocking   & 90\% (4940) & 85\% (4941) & 66\% (4944) & 55\% (4952) \\
ServiceLoader\_MultipleConcurrentThreads & 85\% (4940) & 81\% (4961) & 54\% (4937) & 50\% (4934) \\
Socket\_CloseInput                       & 75\% (4602) & 70\% (4568) & 38\% (4586) & 32\% (4583) \\
Socket\_InputStreamUnavailable           & 90\% (4904) & 85\% (4912) & 68\% (4911) & 57\% (4912) \\
Socket\_LargeReceiveBuffer               & 80\% (4754) & 76\% (4795) & 48\% (4768) & 40\% (4757) \\
Socket\_ReuseAddress                     & 80\% (4732) & 76\% (4765) & 48\% (4746) & 41\% (4730) \\
Thread\_SetDaemonBeforeStart             & 95\% (4939) & 90\% (4939) & 78\% (4941) & 69\% (4950) \\
Throwable\_InitCauseOnce                 & 78\% (3971) & 73\% (3937) & 42\% (3925) & 34\% (3948) \\
\hline
\end{tabularx}
\end{table}

%% file: tikzpictures/disable-probability-and-length-events-processed.tex
\begin{table}[H]
\centering
\begin{tabular}{|ccc|}
\hline
$\rho$ & $\eta$ & \%     \\
\hline
0.1  & 3      & 84.9\% \\
0.1  & 6      & 77.1\% \\
0.3  & 3      & 65.7\% \\
0.3  & 6      & 54.2\% \\
\hline
\end{tabular}
\end{table}

%% file: approximate-alternate-monitors.tex
\section{Approximate Alternate Monitors}
\label{sec:approximate-alternate-monitors}

We've already discussed the structure of an \emph{optimal} alternate monitor. We now move the discussion to non-optimal alternate monitors. These may be desirable due to variety of reasons -- smaller number of states, or a better tradeoff between violations reported and overhead incurred.

For a primary-alternate optimal pair $(\phi, \psi^*)_R$, the number of states in $\psi^*$ may be exponential in $\abs{Q}$ after determinization (up to $2^{|Q^\psi|-1}$, see \cref{remark:num-states} below). In our own empirical evaluation in the previous section, all properties had 5 or fewer states in their minimized DFA form. While we observed the size of most properties being considered in recent literature to be small  (8 states or less \cite{Pradel10}), the properties for monitoring are allowed to be specified by the user and hence, can have arbitrary size. Moreover, several properties specified by the user may be combined into a single property to be monitored \cite{Purandare13} which can have a large size.

For properties with a large number of states, it is desirable to have alternate monitors of size polynomial in $|Q|$. The problem is related to finding closest over-approximation of a regular language within $n$ states, which is conjectured to be hard \cite{Gange17ATVA}. There is already a line of work \cite{Gange17ATVA,Luchaup14INFOCOM} on over-approximation of DFAs and NFAs which we've detailed in our related work section. While we do not present or evaluate any algorithms, in this section we consider an important property of alternate monitors that can aid development of heuristics for the construction of such approximate monitors.

\begin{remark}[Number of states in optimal determinized alternate monitor]
\label{remark:num-states}
Consider a primary-alternate optimal pair $(\phi, \psi^*)_R$. $\psi^*$ is a NFA with $\abs{Q}$ states. A NFA with $n$ states may have upto $2^n$ states after determinization. But as \cref{lemma:err-partition-refinement} below states, $S \subseteq Q^\psi \setminus q_{err}$ and $S$ are mergable, so determinization of $\psi^*$ may have only upto $2^{n-1}$, i.e. $2^{\abs{Q^\psi} - 1}$ states.
\end{remark}

\begin{definition}[Partition refinements]
If $\mathcal{P}$ and $\mathcal{Q}$ are partitions of a set $S$, $\mathcal{Q}$ is called a \emph{coarsening} of $\mathcal{P}$ and $\mathcal{P}$ is called a \emph{refinement} of $\mathcal{Q}$ iff $ \forall P \in \mathcal{P} ~\exists Q \in \mathcal{Q} \sst P \subseteq Q$.
\end{definition}

\begin{lemma}
\label{lemma:err-partition-refinement}
$\mathcal{P}^{err}(\psi)$ is defined as a partition on $2^{Q^\psi}$ such that its classes contain exactly two elements -- $S \subseteq Q \setminus q_{err}$ and $S \cup \tset{q_{err}}$, i.e. $\mathcal{P}^{err}(\psi) = \tset{\tset{S, S \cup \tset{q_{err}}} \mid S \subseteq Q^\psi \setminus q_{err}}$. 
$\mathcal{P}^{err}(\psi)$ is a refinement of the partitioning in DFA minimization of a determinized NFA, i.e. $S$ and $S \cup \set{q_{err}}$ are merged together into a single state in the minimum-state determinization of a NFA. \proofinappendix
\end{lemma}

Our strategy for constructing these approximate monitors is to omit some states in the determinized output. In order to eliminate these states, we need to answer the question of what to do with the incoming transitions to these states. It turns out that we can redirect the transitions to certain other states without losing completeness. We prove this in the following lemma.

\begin{figure}[t]
\setcounter{subfigure}{0}
\centering
\begin{subfigure}{0.48\textwidth}
    \input{tikzpictures/nfa-optimal.tex}
    \caption{Artifical NFA property}
\end{subfigure}
\begin{subfigure}{0.48\textwidth}
    \input{tikzpictures/nfa-nonoptimal.tex}
    \caption{Approximate DFA property}
\end{subfigure}
\caption{Approximate alternate monitors (For brevity, ${q_i\ldots q_j}$ is written as ${i\ldots
s j}$)}
\label{fig:approximation-example}
\end{figure}
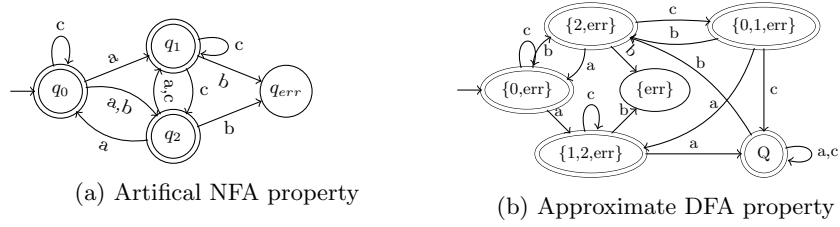

\begin{lemma}
\label{lemma:transition-replacement}
For a primary-alternate pair $(\phi, \psi)_R$ where $\psi$ is a superposed monitor's property, if we update $\delta_\psi(S, \gamma) \gets S'$ where $S' \supseteq \delta_\psi(S, \gamma)$ to obtain $\psi'$, then $(\phi, \psi')_R$ is a superposed primary-alternate pair.
\end{lemma}
\begin{proof}
First, $Q_{\psi'} = Q_{\psi}$ and thus the states of $\psi'$ can be labelled by subsets of $Q_{\phi}$. We have to only show that the superposed monitor condition holds. We induct on length of $f(x)$.

\textbf{Base Case:} $\abs{f(x)} = 0$. $M_{\psi'}$ is in $q_{0}$ and $M_\phi$ is in $q_0$.\\
\textbf{Induction Hypothesis:} For $\abs{f(x)} = n$, $\delta^\phi(q_0, x) \in \delta^{\psi'}(\set{q_0}, f(x))$. \\
\textbf{Induction Step:} Consider a observed lossy stream of length $n + 1$: $f(x \cdot a) =  y \cdot \gamma, \gamma \in R(a), \abs{y} = n$. 
Let $q = \delta^\phi(q_0, x)$, and $q' = \delta^\phi(q_0, x \cdot a)$. We have to show that $q' \in \delta^{\psi'}(q_0, y \cdot \gamma)$.
\begin{align*}
          & q'\in \delta^{\psi}(q, \gamma) \tag{$\because M^\psi$ is superposed}\\
          & \delta^{\psi}(q, \gamma) \subseteq \delta^{\psi'}(q, \gamma) \tag{by construction} \\
 \implies & q'\in \delta^{\psi'}(q, \gamma) \tag{1} \\
          & q \in \delta^{\psi'}(q_0, y) \tag{2: from IH}  \\ 
 \implies & q' \in \delta^{\psi'}(\delta^{\psi'}(q_0, y), \gamma) \tag{from 1 and 2} \\
 \implies & q' \in \delta^{\psi'}(q_0, y \cdot \gamma) \tag*{$\qed$}
\end{align*}%
\end{proof}

Using the \autoref{lemma:transition-replacement}, we can discard a state $S$ and redirect all its incoming transitions to another state $S' \in 2^Q,\ S' \supseteq S$, while still retaining completeness.

We can now choose $n$ states to keep in $2^Q$ and discard others to construct an alternate monitor with a DFA property that has $n$ states. 

We show an example in \cref{fig:approximation-example} of a NFA's approximate determinization. This approximate monitor is on 6 states, whereas the complete minimum-state determinization of \cref{fig:approximation-example} has 8 states and is able to report violations on more strings ($e.g.~``bcb"$). Still, the approximate monitor does not lose the error reporting ability and neither does it lose completeness.

%% file: tikzpictures/nfa-optimal.tex
\begin{tikzpicture}[node distance=1.2,scale=0.8,on grid] %
  \tikzset{every node/.style={scale=0.8}}

  \node[state,accepting,initial]    (q_0)                           {$q_0$}; 
  \node[state,accepting]            (q_1)     [above right=0.6 and 1.5 of q_0]  {$q_1$}; 
  \node[state]                      (q_{err}) [right= 3 of q_0]  {$q_{err}$}; 
  \node[state,accepting]            (q_2)     [below= 1.2 of q_1]   {$q_2$}; 
  
  \path[->] (q_0)     edge                node [above,sloped] {a} (q_1)
            (q_0)     edge  [bend left]   node [below,sloped] {a,b} (q_2)
            (q_0)     edge  [loop above]  node [above] {c} ()
            
            (q_1)     edge  [loop right]  node [right] {c} ()
            (q_1)     edge                    node [below,sloped] {b}   (q_{err})
            (q_1)     edge  [bend left]   node [right] {c}   (q_2)
            
            (q_2)     edge  [bend left]   node [above,sloped]  {a,c}   (q_1)
            (q_2)     edge  [bend left]   node [below,sloped] {a}   (q_0)
            (q_2)     edge                node [below]  {b}   (q_{err});
\end{tikzpicture}

%% file: tikzpictures/nfa-nonoptimal.tex
\begin{tikzpicture}[node distance=1.2,scale=0.8,on grid] %
  \tikzset{every node/.style={scale=0.7}}

  \tikzset{every state/.style={shape=ellipse,very thin}}
  \node[state,accepting,initial]    (0)                         {\tset{0,err}}; 
  \node[state,accepting]            (2)   [above right = of 0]  {\tset{2,err}}; 
  \node[state,accepting]            (01)  [right = 2.3 of 2]        {\tset{0,1,err}}; 
  \node[state,accepting]            (12)  [below right = of 0]    {\tset{1,2,err}}; 
  \node[state,accepting]            (Q) [right = 2.3 of 12]        {Q}; 
  \node[state]                      (err) [right=1.7 of 0]       {\tset{err}}; 

  \path[->] (0)     edge  [bend left]   node [right] {b} (2)
            (0)     edge  [loop above]  node [above] {c} ()
            (0)     edge                node [above] {a} (12)
            
            (2)     edge                node [above,sloped] {b} (err)
            (2)     edge  [bend left=5]   node [above] {c} (01)
            (2)     edge  [bend left]   node [right] {a} (0)
            
            (12)    edge                node [above] {b} (err)
            (12)    edge  [loop above]  node [above] {c} ()
            (12)    edge                node [above] {a} (Q)
            
            (01)    edge  [bend left=15]   node [above] {b} (2)
            (01)    edge                node [right] {c} (Q)
            (01)    edge  [bend left]  node [above] {a} (12)
            
            (Q)     edge  [bend right=20]   node [above] {b} (2)
            (Q)     edge  [loop right]  node [right] {a,c} ();
\end{tikzpicture}

%% file: related-work.tex
\section{Related Work}
\label{sec:related}
Runtime monitoring has been an active research area over the past few decades. A significant part of the research in this area has focused on optimizing monitors and controlling the runtime overhead to make monitoring employable in practice. %
Here, we discuss work which is closely related to our approach.

A line of research \cite{Bodden08,Bodden10a,Dwyer07,Dwyer07:AdaptiveOnlinePA} focuses on lossless partial evaluation of the finite state property to build residual monitors which process fewer events during runtime. \cite{Dwyer07} and \cite{Dwyer07:AdaptiveOnlinePA} can be modelled in our framework using loss models where $R^{-1}(y)$ is a singleton set. 

Another line of research \cite{Purandare13,Allabadi18} does not focus directly on reducing the number of events to be processed but proposes purely dynamic optimizations where resources at run-time are constrained. Allabadi et al. \cite{Allabadi18} constrain the number of monitors in a way which retains completeness but loses soundness, and Purandare et al. \cite{Purandare13} combine multiple monitor which share events into a single monitor to reduce the number of monitors updated.

Kauffman et al. \cite{Kauffman19RV} and Joshi et al. \cite{Joshi17} consider monitorability of LTL formulas under losses. \cite{Kauffman19RV} considers natural losses such as loss, corruption, repetition, or out-of-order arrival of an event and gives an algorithm to find monitorability of a LTL formula. They do not construct a monitor, which monitors lossy traces. \cite{Joshi17} considers monitorability of LTL formulas in the presence of one loss type which is equivalent to our dropped-count filter in \cref{fig:dropped-count-filter} with $n = 1$. They only handle the formulas whose synthesized monitor has transitions, which always lead to just one state, and their construction is only able to recover from losses when it observes such a transition. In general, recurrence temporal properties \cite{Manna90PODC} that can be modeled by B\"{u}chi automata are naturally immune to event losses due to loops in their structures. Our work primarily focuses on safety properties.   %

Falzon et al. \cite{Falzon13} consider the construction of an alternate sound monitor when %
for some parts of the traces only aggregate information,
such as the frequency of events but not their order, is available. We formalize this loss type in our framework in \cref{example:frequency-count}. 

Dwyer et al. \cite{Dwyer08} consider sub-properties formed when the alphabet is restricted to its subset to sample sub-properties from a given property. Their construction ensures completeness and is equivalent to our construction with $R = \{(x, y) \mid \allowbreak x \allowbreak \in  \allowbreak \Delta^*y~\forall y \in \Sigma \setminus \Delta\}$, where $\Delta$ is the set of symbols not observed as events. \cref{fig:silent-drop-monitor} generalizes it with  $R = \set{(x, y) \mid x \in \Delta^*y\ \forall y \in \Sigma}$

Basin et al. \cite{Basin13:IncompleteDisagreeingLogs} introduce a 3-valued timed logic to account for missing information in recorded traces for offline analysis. This allows them to report 3 results: if a violation occurred, if it did not occur, or if the knowledge is insufficient to report either. In the problem we consider, instead of having a single representation for missing information we can have multiple representations for different losses which can differ in their power to report an error.

Bartocci et al. \cite{Bartocci13:ARV} introduce statistical methods to inform overhead control and minimize the probability of missing a violation. For the monitors which are disabled, \cite{Stoller12:RVStateEstimation} introduces statistical methods to predict the missing information due to sampling, which is then used in \cite{Bartocci13:ARV} to get a probability that the violation occurred in an incomplete run. Instead of disabling monitoring altogether and predicting missing information, our approach records lossy information about the events to report violations while maintaining completeness.

Recent work considers over-approximation of DFAs and NFAs in the context of network packet inspection \cite{Ceska18TACAS,Luchaup14INFOCOM,Rubin06CCS} and in the general setting \cite{Gange17ATVA}. \cite{Luchaup14INFOCOM} considers keeping only a subset of frequently-visited states and merging the other states into the final state. \cite{Gange17ATVA} formulates the problem of finding an over-approximating DFA as a search problem and provide heuristics to solve it. %
A part of their algorithm involves a NFA to DFA conversion mechanism which merges state $q_1$ into $q_2$ if $L_{q_1} \subseteq L_{q_2}$, which is similar to the operation we perform in \autoref{lemma:transition-replacement}. \cite{Ceska18TACAS} consider over-approximating NFAs by adding self-loops to selected states, which is equivalent to merging them into $q_{acc}$ in the expanded DFA. This leads to any transition going through those paths to be unmonitorable, whereas we merge the states with one of the selected (possibly monitorable) superstates.

%% file: conclusion.tex
\section{Conclusion and Future Work}
\label{sec:conclusion}
In this work, we presented an efficient approach to support finite state monitoring of lossy event streams, where the losses could be natural or artificially induced. Our approach maintains completeness and is optimally sound. In addition to making monitoring feasible under these conditions, the approach should help improve the performance of monitoring enabling its deployment in production environments. We provide efficient methods to construct optimal monitors automatically from property specifications. We provide an example of how this can be extended in the future to construct approximate alternate monitors for larger properties. We hope that this novel approach will make monitoring particularly attractive under in the presence of high-frequency events and lossy channels. %
In the future, we would like to extend our framework to address infinite state monitors and empirically compare various loss types. %

%% file: appendix.tex
\section{Appendix}

\subsection{\cref{sec:basic-definitions}}
\begin{proof}[\autoref{lemma:well-defined-NFA-labels}]
We refer to the DFA minimization algorithm in \cite{sipser13} which works by 1) constructing an undirected graph $G$ of states which cannot be merged together and 2) constructing classes of states which will be merged together.

Suppose $S_1 \in [S_2]$. To prove: $S_1 \cup S_2 \in [S_2]$. 

We proceed by contradiction. Suppose $S_1 \cup S_2 \not\in [S_2]$. Then $(S_1 \cup S_2,~S_2) \not\in G$.

This implies $\exists x \in \Sigma^*$ such that one of $\delta(S_1,~x)$ and $\delta(S_1 \cup S_2,~x)$ is a final state in the DFA and the other isn't.

\paragraph{\textbf{Case 1:} $\delta(S_1 \cup S_2, x) \cap F = \phi \land \delta(S_2, x) \cap F \neq \phi$.} 
$\delta(S_1 \cup S_2,~x) = \delta(S_1, x) \cup \delta(S_2,~x) \cap F \implies F \cap \delta(S_1,~x) \neq \phi \land F \cap \delta(S_2,~x) \neq \phi$, a contradiction.

\paragraph{\textbf{Case 2:} $\delta(S_1 \cup S_2, x) \cap F \neq \phi \land \delta(S_2, x) \cap F = \phi$.}

$\exists q \in F$ s.t. $ q \in \delta(S_1 \cup S_2, x) \implies q \in \delta(S_1, x) \cup \delta(S_2, x) \implies q \in \delta(S_1, x) $. 

This implies $\delta(S_1, x) \cap F \neq \phi \implies (S_1, S_2) \in G$, which is a contradiction.

\end{proof}

\subsection{\cref{sec:filters}}

\begin{proof}[\cref{remark:polytime-rinv}] $\delta(q, R^{-1}(y))$ is polynomial time computable for R represented as a NFT]
If R is represented by a NFT, then $R^{-1}(y)$ is a regex and $\delta(q_1, R^{-1}(y)) = S$ can be computed in polynomial time: $\forall~q_2 \in Q$ the intersection of $R^{-1}(y)$ and the regex formed by the set of strings which go from $q_1$ to $q_2$ is nonempty then $q_2 \in S$. We loop over $O(n)$ states and check if each is in $S$, and in each iteration the intersection and checking non-emptiness is polynomial time.
\end{proof}

\begin{definition}[$L_{sopt}(\phi, R)$]
For a property $\phi$ and loss model $R$, $L_{sopt}(\phi, R) \defeq \FR{R}(L(\phi)^C)$, i.e. $L_{sopt}$ is the set of lossy strings in $\Gamma^*$ produced by a error execution in $\Sigma^*$. $L_{sopt}$ is the smallest set of strings on which a sound alternate monitor cannot reach a true verdict. 
\end{definition}

\begin{proof}{(\cref{remark:sound-alternate-monitors})}
We argue that the construction in \cref{remark:sound-alternate-monitors} recognizes $L_{sopt}$. It is based on the superposed monitor construction in \cref{theorem:optimality}. Since superposed monitor guarantees that if original monitor is in state $q_{err}$, alternate monitor current state $S$ contains $q_{err}$ and hence we will report the string as a violation. Therefore the construction is sound.

Since the construction in \cref{theorem:optimality} is the minimum set of states we must be in to monitor while maintaining completeness, it is guaranteed that $\forall~q \in S$ for the current alternate monitor state $S$, $\exists x \in \Sigma^*$ and $f \in \FR{R}$ such that $\delta^\phi(q_{0}, x) = q$ and $q \in \delta^\psi(q_{0}, f(x))$. Therefore we only error on strings present in $L_{sopt}$.
\end{proof}

\subsection{\cref{sec:approximate-alternate-monitors}}

\begin{proof}{\cref{lemma:err-partition-refinement}}
The only nonaccept state in the determinization of an alternate NFA is $\set{q_{err}}$ .

Now for any string $x \in \Sigma*$, either $\delta(S, x) = \set{q_{err}} = \delta(S \cup \set{q_{err}}, x)$ or $\delta(q, x) \neq \set{q_{err}} \neq \delta(S \cup \set{q_{err}}, x)$, i.e. either both end up in an accept state, or both in a nonaccept state. 
\end{proof}

\begin{figure}
    \hspace{-0.45\textwidth}\includegraphics[scale=0.45]{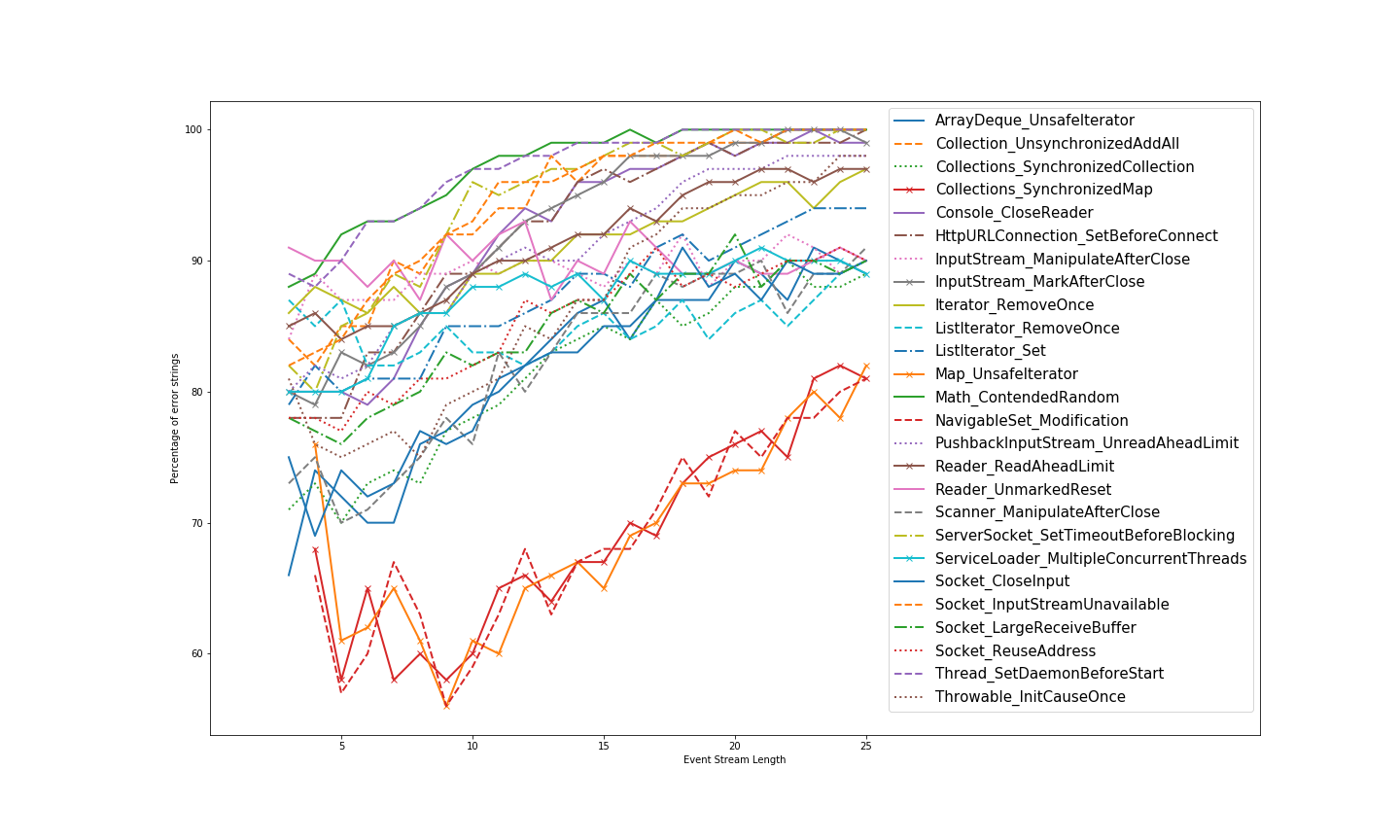}
    \caption{String length vs percentage of reported violations, for all properties}
    \label{fig:chart-all-properties}
\end{figure}

\input{tikzpictures/eval-table-appendix.tex}

%% file: tikzpictures/eval-table-appendix.tex
\newcolumntype{Y}{>{\centering\arraybackslash}X} %
\scriptsize

\begin{table}[t]
\scriptsize
\caption{Error detection for medium strings (Length 10 to 15).}
\vspace{2mm}
\begin{tabularx}{\textwidth}{|l|Y|Y|Y|Y|}
\hline
\textbf{Property}                        & $\rho: 0.1, \eta: 3$     & $\rho: 0.1, \eta: 6$     & $\rho: 0.3, \eta: 3$ &  $\rho: 0.3, \eta: 6$    \\
\hline
ArrayDeque\_UnsafeIterator               & 84\% (4986) & 78\% (4971) & 51\% (4978) & 39\% (4981) \\
Collections\_SynchronizedCollection      & 82\% (4910) & 75\% (4915) & 48\% (4924) & 39\% (4918) \\
Collections\_SynchronizedMap             & 66\% (3702) & 55\% (3717) & 23\% (3728) & 17\% (3660) \\
Collection\_UnsynchronizedAddAll         & 96\% (4998) & 90\% (5000) & 76\% (4999) & 65\% (4999) \\
Console\_CloseReader                     & 94\% (4991) & 88\% (4989) & 73\% (4989) & 58\% (4985) \\
HttpURLConnection\_SetBeforeConnect      & 94\% (4996) & 87\% (4993) & 74\% (4988) & 58\% (4988) \\
InputStream\_MarkAfterClose              & 94\% (4989) & 87\% (4985) & 74\% (4990) & 60\% (4987) \\
Iterator\_RemoveOnce                     & 91\% (4757) & 86\% (4757) & 67\% (4744) & 60\% (4774) \\
ListIterator\_RemoveOnce                 & 84\% (3853) & 79\% (3813) & 55\% (3812) & 50\% (3855) \\
ListIterator\_Set                        & 87\% (4622) & 80\% (4602) & 57\% (4586) & 46\% (4588) \\
List\_UnsynchronizedSubList              & 83\% (4986) & 77\% (4978) & 51\% (4975) & 40\% (4979) \\
Map\_UnsafeIterator                      & 65\% (3673) & 56\% (3717) & 23\% (3635) & 16\% (3683) \\
Math\_ContendedRandom                    & 98\% (5000) & 95\% (4998) & 91\% (4998) & 81\% (5000) \\
NavigableSet\_Modification               & 66\% (3687) & 59\% (3687) & 23\% (3697) & 16\% (3691) \\
PushbackInputStream\_UnreadAheadLimit    & 91\% (4852) & 84\% (4842) & 70\% (4819) & 56\% (4832) \\
Reader\_ReadAheadLimit                   & 91\% (4709) & 85\% (4698) & 70\% (4713) & 58\% (4699) \\
Reader\_UnmarkedReset                    & 90\% (2583) & 89\% (2519) & 70\% (2547) & 70\% (2510) \\
Scanner\_ManipulateAfterClose            & 84\% (4985) & 77\% (4979) & 50\% (4980) & 39\% (4979) \\
ServerSocket\_SetTimeoutBeforeBlocking   & 97\% (5000) & 91\% (4999) & 81\% (4999) & 67\% (4997) \\
ServiceLoader\_MultipleConcurrentThreads & 88\% (5000) & 85\% (4996) & 63\% (4996) & 56\% (5000) \\
Socket\_CloseInput                       & 83\% (4978) & 77\% (4984) & 50\% (4974) & 40\% (4974) \\
Socket\_InputStreamUnavailable           & 97\% (5000) & 91\% (4998) & 82\% (4996) & 69\% (5000) \\
Socket\_LargeReceiveBuffer               & 85\% (4994) & 82\% (4982) & 55\% (4988) & 47\% (4992) \\
Socket\_ReuseAddress                     & 86\% (4986) & 80\% (4990) & 57\% (4984) & 48\% (4988) \\
Thread\_SetDaemonBeforeStart             & 98\% (4999) & 94\% (4999) & 89\% (4998) & 79\% (4998) \\
Throwable\_InitCauseOnce                 & 85\% (4811) & 77\% (4778) & 53\% (4781) & 39\% (4785) \\
\hline
\end{tabularx}
\end{table}

\begin{table}[t]
\scriptsize
\caption{\protect\centering Error detection for large strings (Length 15 to 20). Only the properties which still have at least one non-violating string shown.}
\vspace{2mm}
\begin{tabularx}{\textwidth}{|l|Y|Y|Y|Y|}
\hline
\textbf{Property}                        & $\rho: 0.1, \eta: 3$     & $\rho: 0.1, \eta: 6$     & $\rho: 0.3, \eta: 3$ &  $\rho: 0.3, \eta: 6$    \\
\hline
ArrayDeque\_UnsafeIterator            & 87\% (4999) & 83\% (5000) & 58\% (4999) & 48\% (5000) \\
Collections\_SynchronizedCollection   & 86\% (4987) & 81\% (4987) & 55\% (4992) & 45\% (4989) \\
Collections\_SynchronizedMap          & 73\% (4572) & 61\% (4547) & 29\% (4566) & 21\% (4575) \\
Console\_CloseReader                  & 98\% (5000) & 92\% (4999) & 86\% (5000) & 70\% (5000) \\
HttpURLConnection\_SetBeforeConnect   & 98\% (4998) & 92\% (4999) & 86\% (4998) & 71\% (5000) \\
InputStream\_MarkAfterClose           & 98\% (4999) & 93\% (5000) & 85\% (5000) & 70\% (5000) \\
Iterator\_RemoveOnce                  & 93\% (4915) & 89\% (4915) & 73\% (4905) & 65\% (4913) \\
ListIterator\_RemoveOnce              & 85\% (4263) & 80\% (4252) & 59\% (4267) & 50\% (4261) \\
ListIterator\_Set                     & 90\% (4851) & 84\% (4854) & 64\% (4840) & 53\% (4851) \\
List\_UnsynchronizedSubList           & 88\% (4998) & 83\% (4998) & 59\% (4999) & 48\% (4998) \\
Map\_UnsafeIterator                   & 72\% (4542) & 63\% (4531) & 30\% (4519) & 21\% (4574) \\
NavigableSet\_Modification            & 73\% (4552) & 62\% (4560) & 32\% (4569) & 21\% (4553) \\
PushbackInputStream\_UnreadAheadLimit & 95\% (4966) & 88\% (4949) & 80\% (4953) & 66\% (4965) \\
Reader\_ReadAheadLimit                & 95\% (4903) & 88\% (4910) & 79\% (4914) & 66\% (4902) \\
Reader\_UnmarkedReset                 & 90\% (2458) & 90\% (2549) & 69\% (2483) & 71\% (2497) \\
Scanner\_ManipulateAfterClose         & 88\% (5000) & 83\% (4999) & 58\% (4999) & 48\% (4999) \\
Socket\_CloseInput                    & 88\% (5000) & 82\% (4999) & 59\% (5000) & 48\% (4999) \\
Socket\_LargeReceiveBuffer            & 89\% (4999) & 85\% (4999) & 61\% (5000) & 55\% (5000) \\
Socket\_ReuseAddress                  & 89\% (5000) & 84\% (4999) & 63\% (4999) & 54\% (4999) \\
Throwable\_InitCauseOnce              & 93\% (4953) & 84\% (4966) & 66\% (4959) & 50\% (4955) \\
\hline
\end{tabularx}
\end{table}